\documentclass[pra,aps,showpacs,onecolumn,twoside,superscriptaddress]{revtex4}

\usepackage{mathtools}
\usepackage{amsmath,amsfonts,amssymb,caption,color,epsfig,graphics,graphicx,hyperref,latexsym,mathrsfs,revsymb,theorem,url,verbatim,epstopdf,float}
\usepackage{tikz}
\usepackage[level]{datetime}
\usetikzlibrary{shapes, arrows}

\hypersetup{colorlinks,linkcolor={blue},citecolor={red},urlcolor={blue}}

\newtheorem{definition}{Definition}
\newtheorem{proposition}[definition]{Proposition}
\newtheorem{lemma}[definition]{Lemma}

\newtheorem{theorem}[definition]{Theorem}
\newtheorem{corollary}[definition]{Corollary}
\newtheorem{conjecture}[definition]{Conjecture}

\newtheorem{remark}[definition]{Remark}
\newtheorem{example}[definition]{Example}
\newtheorem{question}[definition]{Question}

\def\bcj{\begin{conjecture}}
\def\ecj{\end{conjecture}}
\def\bcr{\begin{corollary}}
\def\ecr{\end{corollary}}
\def\bd{\begin{definition}}
\def\ed{\end{definition}}
\def\bea{\begin{eqnarray}}
\def\eea{\end{eqnarray}}
\def\bem{\begin{enumerate}}
\def\eem{\end{enumerate}}
\def\bex{\begin{example}}
\def\eex{\end{example}}
\def\bim{\begin{itemize}}
\def\eim{\end{itemize}}
\def\bl{\begin{lemma}}
\def\el{\end{lemma}}
\def\bma{\begin{bmatrix}}
\def\ema{\end{bmatrix}}
\def\bpf{\begin{proof}}
\def\epf{\end{proof}}
\def\bpp{\begin{proposition}}
\def\epp{\end{proposition}}
\def\bqu{\begin{question}}
\def\equ{\end{question}}
\def\br{\begin{remark}}
\def\er{\end{remark}}
\def\bt{\begin{theorem}}
\def\et{\end{theorem}}

\def\squareforqed{\hbox{\rlap{$\sqcap$}$\sqcup$}}
\def\qed{\ifmmode\squareforqed\else{\unskip\nobreak\hfil
\penalty50\hskip1em\null\nobreak\hfil\squareforqed
\parfillskip=0pt\finalhyphendemerits=0\endgraf}\fi}
\def\endenv{\ifmmode\;\else{\unskip\nobreak\hfil
\penalty50\hskip1em\null\nobreak\hfil\;
\parfillskip=0pt\finalhyphendemerits=0\endgraf}\fi}
\newenvironment{proof}{\noindent \textbf{{Proof.~} }}{\qed}
\def\Dbar{\leavevmode\lower.6ex\hbox to 0pt
{\hskip-.23ex\accent"16\hss}D}
\makeatletter
\def\url@leostyle{%
  \@ifundefined{selectfont}{\def\UrlFont{\sf}}{\def\UrlFont{\small\ttfamily}}}
\makeatother
\def\bcj{\begin{conjecture}}
\def\ecj{\end{conjecture}}
\def\bcr{\begin{corollary}}
\def\ecr{\end{corollary}}
\def\bd{\begin{definition}}
\def\ed{\end{definition}}
\def\bea{\begin{eqnarray}}
\def\eea{\end{eqnarray}}
\def\bem{\begin{enumerate}}
\def\eem{\end{enumerate}}
\def\bex{\begin{example}}
\def\eex{\end{example}}
\def\bim{\begin{itemize}}
\def\eim{\end{itemize}}
\def\bl{\begin{lemma}}
\def\el{\end{lemma}}
\def\bpf{\begin{proof}}
\def\epf{\end{proof}}
\def\bpp{\begin{proposition}}
\def\epp{\end{proposition}}
\def\bqu{\begin{question}}
\def\equ{\end{question}}
\def\br{\begin{remark}}
\def\er{\end{remark}}
\def\bt{\begin{theorem}}
\def\et{\end{theorem}}

\def\btb{\begin{tabular}}
\def\etb{\end{tabular}}

\newcommand{\nc}{\newcommand}

\def\a{\alpha}
\def\b{\beta}

\def\z{\zeta}

\def\r{\rho}
\def\s{\sigma}

\def\G{\Gamma}

 \nc{\bbA}{\mathbb{A}} \nc{\bbB}{\mathbb{B}} \nc{\bbC}{\mathbb{C}}
 \nc{\bbD}{\mathbb{D}} \nc{\bbE}{\mathbb{E}} \nc{\bbF}{\mathbb{F}}
 \nc{\bbG}{\mathbb{G}} \nc{\bbH}{\mathbb{H}} \nc{\bbI}{\mathbb{I}}
 \nc{\bbJ}{\mathbb{J}} \nc{\bbK}{\mathbb{K}} \nc{\bbL}{\mathbb{L}}
 \nc{\bbM}{\mathbb{M}} \nc{\bbN}{\mathbb{N}} \nc{\bbO}{\mathbb{O}}
 \nc{\bbP}{\mathbb{P}} \nc{\bbQ}{\mathbb{Q}} \nc{\bbR}{\mathbb{R}}
 \nc{\bbS}{\mathbb{S}} \nc{\bbT}{\mathbb{T}} \nc{\bbU}{\mathbb{U}}
 \nc{\bbV}{\mathbb{V}} \nc{\bbW}{\mathbb{W}} \nc{\bbX}{\mathbb{X}}
 \nc{\bbZ}{\mathbb{Z}}

 \nc{\bA}{{\bf A}} \nc{\bB}{{\bf B}} \nc{\bC}{{\bf C}}
 \nc{\bD}{{\bf D}} \nc{\bE}{{\bf E}} \nc{\bF}{{\bf F}}
 \nc{\bG}{{\bf G}} \nc{\bH}{{\bf H}} \nc{\bI}{{\bf I}}
 \nc{\bJ}{{\bf J}} \nc{\bK}{{\bf K}} \nc{\bL}{{\bf L}}
 \nc{\bM}{{\bf M}} \nc{\bN}{{\bf N}} \nc{\bO}{{\bf O}}
 \nc{\bP}{{\bf P}} \nc{\bQ}{{\bf Q}} \nc{\bR}{{\bf R}}
 \nc{\bS}{{\bf S}} \nc{\bT}{{\bf T}} \nc{\bU}{{\bf U}}
 \nc{\bV}{{\bf V}} \nc{\bW}{{\bf W}} \nc{\bX}{{\bf X}}
 \nc{\bZ}{{\bf Z}}

\nc{\cA}{{\cal A}} \nc{\cB}{{\cal B}} \nc{\cC}{{\cal C}}
\nc{\cD}{{\cal D}} \nc{\cE}{{\cal E}} \nc{\cF}{{\cal F}}
\nc{\cG}{{\cal G}} \nc{\cH}{{\cal H}} \nc{\cI}{{\cal I}}
\nc{\cJ}{{\cal J}} \nc{\cK}{{\cal K}} \nc{\cL}{{\cal L}}
\nc{\cM}{{\cal M}} \nc{\cN}{{\cal N}} \nc{\cO}{{\cal O}}
\nc{\cP}{{\cal P}} \nc{\cQ}{{\cal Q}} \nc{\cR}{{\cal R}}
\nc{\cS}{{\cal S}} \nc{\cT}{{\cal T}} \nc{\cU}{{\cal U}}
\nc{\cV}{{\cal V}} \nc{\cW}{{\cal W}} \nc{\cX}{{\cal X}}
\nc{\cZ}{{\cal Z}}

\nc{\hA}{{\hat{A}}} \nc{\hB}{{\hat{B}}} \nc{\hC}{{\hat{C}}}
\nc{\hD}{{\hat{D}}} \nc{\hE}{{\hat{E}}} \nc{\hF}{{\hat{F}}}
\nc{\hG}{{\hat{G}}} \nc{\hH}{{\hat{H}}} \nc{\hI}{{\hat{I}}}
\nc{\hJ}{{\hat{J}}} \nc{\hK}{{\hat{K}}} \nc{\hL}{{\hat{L}}}
\nc{\hM}{{\hat{M}}} \nc{\hN}{{\hat{N}}} \nc{\hO}{{\hat{O}}}
\nc{\hP}{{\hat{P}}} \nc{\hR}{{\hat{R}}} \nc{\hS}{{\hat{S}}}
\nc{\hT}{{\hat{T}}} \nc{\hU}{{\hat{U}}} \nc{\hV}{{\hat{V}}}
\nc{\hW}{{\hat{W}}} \nc{\hX}{{\hat{X}}} \nc{\hZ}{{\hat{Z}}}

\nc{\hn}{{\hat{n}}}

\def\dim{\mathop{\rm Dim}}

\def\lin{\mathop{\rm span}}

\def\rank{\mathop{\rm rank}}

\def\tr{\mathop{\rm Tr}}

\def\ox{\otimes}

\newcommand{\bra}[1]{\langle#1|}
\newcommand{\ket}[1]{|#1\rangle}
\newcommand{\proj}[1]{| #1\rangle\!\langle #1 |}

\def\Dbar{\leavevmode\lower.6ex\hbox to 0pt
{\hskip-.23ex\accent"16\hss}D}
\usepackage{color}
\begin{document}
\title{Inertia of two-qutrit entanglement witnesses}

\newdateformat{ukdate}{\ordinaldate{\THEDAY} \monthname[\THEMONTH] \THEYEAR}
\date{\ukdate\today}


\author{Changchun Feng}
\affiliation{LMIB(Beihang University), Ministry of Education, and School of Mathematical Sciences, Beihang University, Beijing 100191, China}

\author{Lin Chen}\email[]{linchen@buaa.edu.cn (corresponding author)}
\affiliation{LMIB(Beihang University), Ministry of Education, and School of Mathematical Sciences, Beihang University, Beijing 100191, China}
\affiliation{International Research Institute for Multidisciplinary Science, Beihang University, Beijing 100191, China}

\author{Chang Xu} 
\email[]{20191209@stu.neu.edu.cn}
\affiliation{Science College, Northeastern University, Shenyang, 110000, China}

\author{Yi Shen}
\affiliation{LMIB(Beihang University), Ministry of Education, and School of Mathematical Sciences, Beihang University, Beijing 100191, China}
\affiliation{School of Science, Jiangnan University, Wuxi 214122, China}

\begin{abstract}

Entanglement witnesses (EWs) are a fundamental tool for the detection of
entanglement. We investigate the inertias of bipartite EWs constructed by the partial transpose of NPT states.  Furthermore, we find out most of the inertias of the partial transpose of the two-qutrit bipartite NPT states. As an application, we extend our results to high dimensional states.

\end{abstract}

\maketitle

\Large

\section{Introduction}

Einstein, Podolsky, Rosen, Schrodinger, et al discovered quantum entanglement. Then it became a remarkable feature of quantum mechanics. It lies in the heart of quantum information theory \cite{2007Quantum,GUHNE20091}. In recent decades, entanglement has been recognized as a kind of valuable resource \cite{2007Quantum,2019Quantum,Patricia2019Resource}. It is widely used in various quantum information processing tasks such as quantum computing \cite{2005Experimental}, teleportation \cite{2004Deterministic}, dense coding \cite{2002Quantum}, cryptography \cite{2020Entanglement}, and quantum key distribution \cite{Xu2020}.

Although several useful separability criteria such as positive-partial-transpose (PPT) criterion \cite{Peres1996,horodecki1997}, range criterion, and realignment criterion \cite{Rudolph2002Some} were developed, all of them cannot strictly distinguish between the set of entangled states and that of separable ones. According to PPT criterion, any state with non positive partial transpose (NPT) must be entangled. Nevertheless,  the converse only holds for two-qubit and qubit-qutrit  systems. It has been shown  that determining whether a bipartite state is entangled is an NP-hard problem \cite{2003Proceedings}. In $2000$, Terhal first introduced the term entanglement witness (EW) by indicating that a violation of a Bell inequality can be expressed as a witness for entanglement \cite{2000Entanglement}. Nowadays, EWs are a fundamental tool for the detection of entanglement both theoretically and experimentally. Actually, many EWs have been implemented with local measurements \cite{2000Entanglement,2019Design, 2020Measurement}.

As we know, the partial transpose of NPT state $\r$ is an EW. The negative eigenvalues of $\r^\G$ are a signature of entanglement. They are closely related to other problems in entanglement theory. The negativity  is a well-known computable entanglement measure \cite{2002Computable}. At the same time, it is the sum of the absolute values of negative eigenvalues of  $\r^\G$. Also, by the definition of $1$-distillable state \cite{Micha1998Mixed}, when $\r^\G$ has more negative eigenvalues, $\r$ is  more likely to be $1$-distillable. Thus, it is important to explore the negative eigenvalues of $\r^\G$. The problem of determining how many negative eigenvalues the partial transpose of NPT state could contain has attracted great interest \cite{2008Universal,2012Qubit,2013Negative,2013Non}. It was first specified in \cite{2008Universal} that $\r^\G$ has one negative eigenvalue and three positive eigenvalues for any two-qubit entangled state $\r$. Next the inertias of the partial transpose of $2\times n$ NPT states has been investigated in \cite{2020Inertias}. As far as we know, many quantum information queestions turn into intractable when $2\times n$ systems are replaced by $3\times 3$ systems, such as distillability problem \cite{Micha1998Mixed,1999Evidence}, distinguishable subspace problem \cite{2007On}, and separability problem \cite{2002Distinguishing}. 

Due to Sylvester theorem \cite{2011Matrix}, we introduce a tool to study the inertia of EW, namely the inertia of a Hermitian operator is invariant under SLOCC. Firstly, we introduce what the two-qutrit subspaces contain for subspaces with various dimensions  in Lemma \ref{le:9}. Then we explore a property of the inertia of the partial transpose of bipartite NPT states in Lemma \ref{le:add}. Further we investigate the number of positive eigenvalues of  bipartite EWs in Lemma \ref{le:c=2}. Then we give some observations on the relation between the number of positive eigenvalues and that of negative eigenvalues in Lemma \ref{le:negative}.  Then we propose two observations on the partial transpose of the two-qutrit  NPT states in Lemmas \ref{le:12} and \ref{le:13}, demonstrated by Example \ref{ex:11}. Then we discuss all cases in the partial transpose of the two-qutrit bipartite NPT state in Theorem \ref{th:NPT2}. We draw a conclusion about the inertias of (the partial transpose) of the two-qutrit bipartite NPT states in Theorem \ref{th:N_33}. Next we propose Example \ref{ex:13arrays} to illustrate Theorem \ref{th:N_33}(i). We propose Example \ref{ex:3arrays} to show the relationship between the inertias of $2\times 3$ states and those of $3\times 3$ states.  Finally, we extend some conclusions on the inertias  from  $2\times n$ states to $3\times n$ states in Lemma \ref{le:3n}.

The remainder of this paper is organized as follows. In section \ref{sec:p} we introduce the preliminaries by clarifying the notations and presenting necessary definitions and results. In section \ref{sec:I_1} we focus on the bipartite EWs constructed by (the partial transpose) of NPT state, and determine intertias of such EWs. In section \ref{sec:I_2} we investigate the inertias of the partial transpose of the two-qutrit bipartite NPT states. Finally, we conclude in section \ref{sec:con}. In section \ref{sec:app} we find out the relationship between the inertias of $2\times 3$ states and those of $2\times 3$ states and  extend some conclusions on the inertias  from  $2\times n$ states to $3\times n$ states. We also partially test the existence of two unverified inertias using python program.

\section{Preliminaries}
\label{sec:p}
 In this section we introduce the facts used in this paper. 
 We refer to $\mathbb{C}^{d}$ as the $d$-dimensional Hilbert space. We define $I_n$ as the  identity matrix of order $n$. We denote $\otimes^n_{i=1}\mathbb{C}^{d_i}:= \otimes^n_{i=1}\mathcal{H}_{A_i}$ as the n-partite Hilbert space, where $d_i$'s are local dimensions. Let $\mathbb{M}_n$ be the set of $n\times n$ matrices. If $\rho\in \mathcal{B}(\otimes^n_{i=1}\mathbb{C}^{d_i})$ is a positive semidefinite matrix of trace one, then $\rho$ is an $n$-partite quantum state. The partial transpose w.r.t. system A of a bipartite matrix $\rho \in \mathcal{B}(\mathcal{H}_A\ox \mathcal{H}_B)$ is definied as,
  $\rho^{\Gamma}:=\sum_{i,j}\langle c_i|_A\rho |c_j \rangle \ox |c_j \rangle \langle c_i|_A$,
 where the set of $|c_i \rangle$ is an arbitrary orthonormal basis in $\mathcal{H}_A$. If $\rho^{\Gamma}$ has at least one negative eigenvalue then $\rho$ has non-positive partial transpose (NPT).  If $\rho^{\Gamma}$ does not have negative eigenvalues then $\rho$ has positive partial transpose (PPT).  We refer to an $M\times N$ state as a bipartite $\rho$ such that rank $\rho_A=M$ and rank $\rho_B=N$. Given the Schmidt decomposition $|\psi\rangle =\sum_i\sqrt{c_i}|a_i,b_i\rangle$, we refer to the number of nonzero $c_i$ as the Schmidt rank of $\psi$. We denote the  number as SR($\psi$). There is an essential proposition for the matrix inertia, namely Sylvester theorem \cite{Horn1985}. It states that Hermitian matrices $A, B \in \bbM^n$
 have the same inertia if and only if there is a non-singular matrix $S$ such that $B = SAS^\dagger$. 
 
 In the following we review the entanglement witness (EW), inertia and stochastic local operations and classical communications (SLOCC) equivalence.

 \begin{definition}
 	\label{de:EW}
 	Suppose $W\in \mathcal{B}(\otimes^n_{i=1}\mathbb{C}^{d_i})$ is Hermitian. We say that W is an $n$-partite EW if 
 	
 	(i) it is non-positive semidefinite, 
 	
 	(ii) $\langle\psi|W|\psi\rangle\geq 0$ for any product vector $|\psi\rangle=\otimes^n_{i=1}|a_i\rangle$ with $|a_i\rangle \in \mathcal{H}_i$.
 \end{definition} 
 
 \begin{definition}
 	\label{de:in}
 	Let $A \in \mathbb{M}_n$ be Hermitian. The inertia of $A$, denoted by $\ln (A)$, is defined as the following sequence 
 	
 	\begin{eqnarray}
 	\label{eq:1}
 		\ln(A):=(v_-,v_0,v_+),
 	\end{eqnarray}
 	where $v_-,v_0$ and $v_+$ are respectively the numbers of negative, zero and positive eigenvalues of $A$.
 	
 	Note that if $\rho$ is PPT, then $v_-=0$.
 \end{definition}
 
 \begin{definition}
 	\label{de:SL}
 	Two $n$-partite pure state $|\alpha \rangle$, $|\beta \rangle$ are SLOCC equivalent if there exists a product invertible operation $Y=Y_1\ox...\ox Y_n$ such that $|\alpha \rangle =Y|\beta \rangle$. 
 \end{definition}
 
 Then we present several results for the state $\rho$ and the EWs \cite{2013Non,2020Inertias,Chen2013,3tensors1983}.
 \begin{lemma} \cite{3tensors1983}
 	\label{le:mn-2}
 	Suppose $\mathcal{V} \subseteq \bbC^m\ox\bbC^n$ is a bipartite Hilbert subspace. 
 	
 	(i) If $\dim(\mathcal{V})=mn-1$, then $\mathcal{V}$ is spanned by product vectors.
 	
 	(ii) If $\dim(\mathcal{V})=mn-2$, then $\mathcal{V}$ is either spanned by product vectors or up to SLOCC equivalence spanned by $\{\ket{0,0}+\ket{1,1},\ket{i,j}:(i,j)\neq(0,0),(0,1),(1,1)\}$.
 \end{lemma}
 
 \begin{lemma} \cite{2013Non}
 	\label{le:pre}
 	Suppose $\rho$ is an $m\times n$ NPT state. If $\rho$ is a pure state with Schmidt rank $r$, then
 	
 	\begin{eqnarray}
 	\ln (\rho^{\Gamma})=(\frac{r^2-r}{2},mn-r^2,\frac{r^2+r}{2}) .
 	\end{eqnarray}
 \end{lemma}

\begin{lemma}\cite[Proposition $6$]{Chen2013}
	\label{le:2}
	Suppose $\mathcal{V} \subseteq \bbC^m \otimes \bbC^n$. If $\dim(\mathcal{V}) >(m-1)(n-1)$ then $\mathcal{V}$ contains at least one product
	vector. Furthermore, if $\dim(\mathcal{V}) >(m-1)(n-1)+1$ then $\mathcal{V}$ has infinitely many product vectors. 
\end{lemma}

\begin{lemma} \cite[Lemma $5$]{2020Inertias}
	\label{le:mn}
	Suppose $W$ is an EW on $\mathbb{C}^m \ox \mathbb{C}^n$.
	
	(i) Let $A$ be the non-positive eigen-space of $W$, i.e., the sum of negative and zero eigen-spaces
	of W. Then the product vectors in $A$ all belong to the zero eigen-space of W. In particular,
	every vector in the negative eigen-space of W is a pure entangled state.
	
	(ii) The number of negative eigenvalues of W is in $[1, (m-1)(n-1)]$. The decomposable EW
	containing exactly $(m-1)(n-1)$ negative eigenvalues exists.

	(iii) The number of positive eigenvalues of W is in $[2, mn-1]$.
\end{lemma}

 We will apply Lemma \ref{le:mn-2} to prove Lemma \ref{le:c=2}. We will use Lemmas \ref{le:pre}, \ref{le:2} and \ref{le:mn} to prove Lemma \ref{th:N_33}. 
Then we denote the inertia set, $\mathcal{N}_{m,n}:=\{\ln (\rho^{\Gamma})|\rho\quad is \quad an\quad m\times n\quad NPT\quad state.\}$. We introduce some observations on the inertias.
\begin{lemma} \cite[Lemma $8$]{2020Inertias}
	\label{le:rel}
	(i) Suppose $\rho$ is an $m\times n$ NPT state and $\rho^{\Gamma}$ has the inertia $(a,b,c)$. Then there is a small enough $x>0$ and NPT state $\sigma:=\rho+xI_{mn}$, such that $\ln(\sigma^{\Gamma})=(a,0,b+c)$. 
	
	(ii) Suppose $m_1\leq m_2$ and $n_1\leq n_2$. If $(a_1,b_1,c_1)\in \mathcal{N}_{m_1,n_1}$, with $a_1+b_1+c_1=m_1n_1$, then for any  $l\in[0,m_2n_2-m_1n_1]$ we have $(a_1, m_2n_2-m_1n_1-l,b_1+c_1+l)\in \mathcal{N}_{m_2,n_2}$.
\end{lemma}


Based on the above preliminary knowledge we are ready to study the inertia of the partial transposes of NPT states.

\section{Inertias of the partial transposes of NPT states}
\label{sec:I_1}
In this section we focus on the bipartite EWs constructed by the partial transpose of NPT state, and determine intertias of such EWs.

At first,   we  introduce what the two-qutrit subspaces contain for subspaces with various dimensions.  
\begin{lemma}
	\label{le:9}
	
	Suppose $\mathcal{V} \subseteq \mathbb{C}^3 \otimes \mathbb{C}^3 $ has infinitely many pairwise linearly independent product vectors.
	
	(i) If the dimension of $\mathcal{V}$ is two,  then $\mathcal{V}$  up to SLOCC equivalence and system permutation contains $\{\ket{0,0},\ket{0,1}\}$.
	
	(ii) If the dimension of $\mathcal{V}$ is three, then $\mathcal{V}$ up to SLOCC equivalence and system permutation contains $\{\ket{0,0},\ket{0,1}\}$ or  $\{\ket{0,0},\ket{1,1}, (\ket{0}+\ket{1})(\ket{0}+\ket{1})\}$.
	
	(iii) If the dimension of $\mathcal{V}$ is four, then $\mathcal{V}$ up to SLOCC equivalence and system permutation contains $\{\ket{0,0},\ket{0,1}\}$ or $\{\ket{0,0},\ket{1,1}, (\ket{0}+\ket{1})(\ket{0}+\ket{1})\}$ or $\{\ket{0,0}, \ket{1,1}, \ket{2}(\ket{0}+\ket{1}), (a\ket{0}+b\ket{1}+\ket{2})(\ket{0}+f\ket{1})\}$ satisfying the conditions $a,b\neq 0$  and $f\neq 0,1$.
	
	(iv)  If the dimension of $\mathcal{V}$ is five, then $\mathcal{V}$ up to SLOCC equivalence and system permutation contains  $\{\ket{0,0},\ket{0,1}\}$ or $\{\ket{0,0},\ket{1,1}, (\ket{0}+\ket{1})(\ket{0}+\ket{1})\}$ or $\{\ket{0,0}, \ket{1,1}, \ket{2}(\ket{0}+\ket{1}), (a\ket{0}+b\ket{1}+\ket{2})(\ket{0}+f\ket{1})\}$, $\{\ket{0,0}, \ket{1,1}, \ket{2,2}, \}$ 
\end{lemma}
\begin{proof}

	(i) If the dimension of $\mathcal{V}$  is two, then we obtain that one product vector is SLOCC equivalent to $\ket{0,0}$. One can verify that the other product vector can only be SLOCC equivalent or system permutated to $\ket{0,1}$.
	
	(ii) Suppose the dimension of $\mathcal{V}$ is three. Since $\mathcal{V}$ has infinitely many product vectors, it may contain the set $\{\ket{0,0},\ket{0,1}\}$ from (i). Suppose $\mathcal{V}$ does not contain the set. It contains product vectors $\{\ket{a_0,b_0},\ket{a_1,b_1},\ket{a_2,b_2}\}$.  Up to SLOCC equivalence, we may assume that $\ket{a_0,b_0}=\ket{0,0}$ and $\ket{a_1,b_1}=\ket{1,1}$. Suppose the third product vector is $(a\ket{0}+b\ket{1}+c\ket{2})(d\ket{0}+f\ket{1}+g\ket{2})$. Recall that $\mathcal{V}$ has infinitely many pairwise linearly independent product vectors. So we have that $x\ket{0,0}+y\ket{1,1}+(a\ket{0}+b\ket{1}+c\ket{2})(d\ket{0}+f\ket{1}+g\ket{2})$ is a product vector and $x,y\neq 0$. Then we obtain that the  rank of  $x\proj{0}+y\proj{1}+(a\ket{0}+b\ket{1}+c\ket{2})(d\bra{0}+f\bra{1}+g\bra{2})= \bma x+ad&af &ag\\ bd&y+bf &bg\\ cd &cf &cg\ema$ is one. Because these three row vectors are pairwise linearly dependent and $x,y\neq 0$, we have $c=f=0$. Then we obtain that $(a\ket{0}+b\ket{1})(d\ket{0}+f\ket{1})$ is SLOCC equivalent to $(\ket{0}+\ket{1})(\ket{0}+\ket{1})$. Then $\mathcal{V}$ up to SLOCC equivalence and system permutation contains $\{\ket{0,0},\ket{1,1}, (\ket{0}+\ket{1})(\ket{0}+\ket{1})\}$.
	
	(iii) Suppose the dimension of $\mathcal{V}$  is four.   Since $\mathcal{V}$ has infinitely many product vectors, it may contain the sets  $\{\ket{0,0},\ket{0,1}\}$ or  $\{\ket{0,0},\ket{1,1}, (\ket{0}+\ket{1})(\ket{0}+\ket{1})\}$ from (ii). Suppose $\mathcal{V}$ does not contain the sets. That is, we hypothesize that the space spanned by any three product vectors of the set does not generate infinitely many product vectors. The sets contains product vectors $\{\ket{a_0,b_0},\ket{a_1,b_1},\ket{a_2,b_2},\ket{a_3,b_3}\}$.  Up to SLOCC equivalence, we may assume that $\ket{a_0,b_0}=\ket{0,0}$ and $\ket{a_1,b_1}=\ket{1,1}$. Then we discuss three cases (iii.A)-(iii.C) in terms of the third product vector $\ket{a_2,b_2}=\ket{\alpha}$ and the fouth product vector $\ket{a_3,b_3}=\ket{\beta}$.
	
	Because of the hypothesis, $\ket{\alpha}$ can only be  $\ket{2}(a\ket{0}+b\ket{1}+c\ket{2})$. We will divide $\ket{\alpha}$ into two cases  up to SLOCC equivalence and permutation. One is $\ket{2,2}$, the other is $\ket{2}(m\ket{0}+n\ket{1})$.  When $\alpha$ is $\ket{2,2}$, we will prove it is impossible  that $\ket{\alpha}$ is $\ket{2,2}$ and $\ket{\beta}$ is $(a\ket{0}+b\ket{1}+c\ket{2})(d\ket{0}+f\ket{1}+g\ket{2})$ in (iii.A). Then we obtain that $\ket{\alpha}$ and $\ket{\beta}$ are not up to SLOCC equivalent to $\ket{2,2}$. Then we consider the other case for  $\ket{\alpha}=\ket{2}(m\ket{0}+n\ket{1})$. If $m=0$ or $n=0$, one can prove that $\{\ket{0,0},\ket{1,1},\ket{\alpha}\}$ generates infinitely many product vectors. We draw a contradiction. So $m,n\neq 0$. We obtain that $\ket{\alpha}$ is SLOCC equivalent to $\ket{2}(\ket{0}+\ket{1})$ in this case. Because $\ket{\beta}$ are not up to SLOCC equivalent to $\ket{2,2}$, we will divide $\ket{\beta}$ into two cases when $\ket{\alpha}=\ket{2}(\ket{0}+\ket{1})$.  Then we will dicuss the case that $\ket{\alpha}$ is $\ket{2}(\ket{0}+\ket{1})$ and $\ket{\beta}$ is $(a\ket{0}+b\ket{1})(d\ket{0}+f\ket{1}+\ket{2})$ in (iii.B), and the case that $\ket{\alpha}$ is $\ket{2}(\ket{0}+\ket{1})$ and $\ket{\beta}$ is $(a\ket{0}+b\ket{1}+\ket{2})(d\ket{0}+f\ket{1})$ in (iii.C). Suppose $d=0$ or $f=0$. $\mathcal{V}$ up to SLOCC equivalence and system permutation contains $\{\ket{0,0},\ket{1,1},(a\ket{0}+b\ket{1}+\ket{2})(d\ket{0}+f\ket{1})\}$. We get that  $\mathcal{V}$ up to SLOCC equivalence and system permutation contains $\{\ket{0,0},\ket{0,1}\}$. We draw a contradiction. Then we have $d,f\neq 0$. We will suppose $d=1$ and $f\neq 0$ in (iii.C).
	
	(iii.A) Suppose $\ket{\alpha}$ is $\ket{2,2}$ and $\ket{\beta}$ is $(a\ket{0}+b\ket{1}+c\ket{2})(d\ket{0}+f\ket{1}+g\ket{2})$.  $\mathcal{V}$ has infinitely many pairwise linearly independent product vectors. Then we obtain that the  rank of  $x\proj{0}+y\proj{1}+z\proj{2}+(a\ket{0}+b\ket{1}+c\ket{2})(d\bra{0}+f\bra{1}+g\bra{2})= \bma x+ad&af &ag\\ bd&y+bf &bg\\ cd &cf &z+cg\ema$ is one. These three row vectors are pairwise linearly dependent.  Because of the hypothesis, we get that $x,y,z\neq 0$. Then one can verify that the rank of the above matrix is more than one. Then we draw a contradiction. 
	
	(iii.B) Suppose $\ket{\alpha}$ is $\ket{2}(\ket{0}+\ket{1})$, $\ket{\beta}$ is $(a\ket{0}+b\ket{1})(d\ket{0}+f\ket{1}+\ket{2})$. Then we obtain that the  rank of  $x\proj{0}+y\proj{1}+z\ket{2}(\bra{0}+\bra{1})+(a\ket{0}+b\ket{1})(d\bra{0}+f\bra{1}+\bra{2})= \bma x+ad&af &a\\ bd&y+bf &b\\z &z &0\ema$ is one. One can verify that if the rank of the  matrix is one, then $a=b=x=y=0$. We draw a contradiction.
	
	(iii.C) Suppose $\ket{\alpha}$ is $\ket{2}(\ket{0}+\ket{1})$, $\ket{\beta}$ is $(a\ket{0}+b\ket{1}+\ket{2})(\ket{0}+f\ket{1})$ and $f\neq 0$. Then we obtain that the  rank of  $x\proj{0}+y\proj{1}+z\ket{2}(\bra{0}+\bra{1})+(a\ket{0}+b\ket{1}+\ket{2})(\bra{0}+f\bra{1})= \bma x+a&af &0\\ b&y+f &0\\ z+1 &z+f &0\ema$ is one. Because of the hypothesis, we get that $x,y,z\neq0$. If $f=1$, then by (ii) we have that $\{\ket{0,0}, \ket{2}(\ket{0}+\ket{1}), (a\ket{0}+b\ket{1}+\ket{2})(\ket{0}+\ket{1})\}$ can generate infinitely many product vectors when $b=0$.  We draw a contradiction. We obtain that $f\neq 1$. Then we consider the values of $a$ and $b$.  Suppose $a=0$. We have that $\{\ket{1}\bra{1},\ket{2}(\bra{0}+\bra{1}),(b\ket{1}+\ket{2})(\bra{0}+f\bra{1})\}$ is a $3$-dimensional subspace. We draw a contradiction. So we have that $a\neq0$. Similarly, we get that $b\neq 0$.  Suppose $a,b\neq 0$ and $f\neq0,1$. We define $n=\dfrac{x+a}{af}$. $\dfrac{x+a}{af}=\dfrac{b}{y+f}=\dfrac{z+1}{z+f}=n$, if and only if the rank of $\bma x+a&af &0\\ b&y+f &0\\ z+1 &z+f &0\ema$ is one. We obtain that $x=afn-a$, $y=\dfrac{b-fn}{n}$  and $z=\dfrac{1-fn}{1-n}$. Because the value of $n$ is arbitrary, there are various combinations of values of $x,y,z\neq 0$.  So we get that $\mathcal{V}$ up to SLOCC equivalence and system permutation contains  $\{\ket{0,0}, \ket{1,1}, \ket{2}(\ket{0}+\ket{1}), (a\ket{0}+b\ket{1}+\ket{2})(d\ket{0}+f\ket{1})\}$ and $a,b\neq 0$ and $f\neq 0,1$.
	
\end{proof}
Then we introduce a property of the inertia of an $n\times n$ bipartite NPT state.
\begin{lemma}
	\label{le:add}
	Suppose $\r$ is an $n\times n$ bipartite NPT state, and $\r^\G=P-Q$ where $P,Q\ge0$ and $P\perp Q$. If $\ket{a,b}\in \cR(\r^\G)$, then $\r^\G+\proj{a,b}$ has the inertia $(\rank Q, n-\rank P-\rank Q, \rank P )$ or $(-1+\rank Q, n-\rank P-\rank Q, 1+\rank P )$ or $(-1+\rank Q,  n+1-\rank P-\rank Q, \rank P )$.   	
\end{lemma}
\begin{proof}
	Because $\ket{a,b}\in \cR(\r^\G)=\cR(P)+\cR(Q)$, we have $\ket{a,b}\in \cR(P)$, $\ket{a,b}\in \cR(Q)$, or $\ket{a,b}\in \{\ket{x},\ket{y}\}$ where $\ket{x}\in \cR(P)$ and $\ket{y}\in \cR(Q)$. If $\ket{a,b}\in \cR(P)$ or $\cR(Q)$ then by the definition of inertia, we obtain that $\r^\G+\proj{a,b}$ has inertia $(\rank Q, n-\rank P-\rank Q, \rank P)$. Next suppose $\ket{a,b}\in \{\ket{x},\ket{y}\}$. We can write up $P=P_1+\proj{x}$ and $Q=Q_1+\proj{y}$ where
	\begin{eqnarray}
	\label{eq:p1q1}
	&& P_1,Q_1\ge0, 
	\notag\\&&
	\rank P=1+\rank P_1, 
	\notag\\&&
	\rank Q=1+\rank Q_1. 	
	\end{eqnarray}
	We have 
	\begin{eqnarray}
	&&\r^\G+\proj{a,b}
	\\=&&	
	P-Q+\proj{a,b}
	\\=&&
	(P_1+\proj{x})-(Q_1+\proj{y})+\proj{a,b}
	\\=&&
	P_1-Q_1+\s,
	\end{eqnarray} 
	where $\s=\proj{x}-\proj{y}+\proj{a,b}$ has rank one or two. Further, \eqref{eq:p1q1} implies that the nonzero vector in $\cR(\s)$ is linearly independent from any nonzero vector in $\cR(P_1)+\cR(Q_1)$. If $\rank\s=1$ then $\s$ is positive semidefinite. So $\r^\G+\proj{a,b}$ has the inertia $(-1+\rank Q, n+1-\rank P-\rank Q, \rank P )$. On the other hand if $\rank\s=2$ then $\s$ has at least one positive eigenvalue. So $\r^\G+\proj{a,b}$ has the inertia $(\rank Q, n-\rank P-\rank Q, \rank P )$ or $(-1+\rank Q, n-\rank P-\rank Q, 1+\rank P )$. We have proven the assertion. 
\end{proof}

 Then we investigate the number of positive eigenvalues of  bipartite EWs.

\begin{lemma}
	\label{le:c=2}
	Every bipartite EW has at least three positive eigenvalues.
\end{lemma}

\begin{proof}
	It is known that every bipartite EW has at least two positive eigenvalues. Thus, we have to show that there is no EW with exact two positive eigenvalues.
	Assume that $W$ is an $m\times n$ EW which has inertia $(a,b,2)$. Then the spectral decomposition of $\rho^\G$ reads as
	\begin{eqnarray}
	\label{eq:sdec-1}
	W=-\sum_{j=1}^a\proj{v_j}+0\cdot\sum_{j=a+1}^{a+b}\proj{v_j}+\proj{v_{mn-1}}+\proj{v_{mn}},
	\end{eqnarray}
	where $\{\ket{v_j}\}_{j=1}^{mn}$ are pairwisely orthogonal.
	It follows from Lemma \ref{le:mn-2} that the non-positive eigenspace of $W$ is either spanned by product vectors or up to SLOCC equivalence spanned by $\{\ket{0,0}+\ket{1,1},\ket{i,j}:(i,j)\neq(0,0),(0,1),(1,1)\}$. 
	
	First, if $\lin\{\ket{v_j}\}_{j=1}^{a+b}$ is spanned by product vectors, there exists a product vector $\ket{x,y}$ which is orthogonal to $\lin\{\ket{v_{mn-1}},\ket{v_{mn}}\}$ but non-orthogonal to $\ket{v_1}$. Hence, we obtain that 
	\begin{eqnarray}
	\label{eq:c=2-1}
	0\leq \bra{x,y}W\ket{x,y}\leq -\bra{x,y}(\proj{v_1})\ket{x,y}<0.
	\end{eqnarray}
	It is a contradiction. 
	
	Second, if $\lin\{\ket{v_j}\}_{j=1}^{a+b}$ is spanned by $\{\ket{0,0}+\ket{1,1},\ket{i,j}:(i,j)\neq(0,0),(0,1),(1,1)\}$ up to SLOCC equivalence. It follows that $\lin\{\ket{v_{mn-1}},\ket{v_{mn}}\}$ is spanned by $\ket{0,0}-\ket{1,1},\ket{0,1}$. It is known that every two-qubit EW has inertia $(1,0,3)$. Hence, we assume one of $m,n$ is greater than two. Suppose one of the local dimensions is two. For example, we may assume $m=2$. Using the projector $(I_2\oplus 0_{n-2})\ox (I_n-(I_2\oplus 0_{n-2}))$ we obtain that
	\begin{eqnarray}
	\label{eq:c=2-2}
	\begin{split}	
	\tilde{W}&:=\big((I_2\oplus 0_{n-2})\ox (I_n-(I_2\oplus 0_{n-2}))\big) W \big((I_2\oplus 0_{n-2})\ox (I_n-(I_2\oplus 0_{n-2}))\big)\\
	&=-\big(I_2\ox (I_n-(I_2\oplus 0_{n-2}))\big)(\sum_{j=1}^a\proj{v_j})\big((I_2\oplus 0_{n-2})\ox (I_n-(I_2\oplus 0_{n-2}))\big).
	\end{split}
	\end{eqnarray}
	It follows that $\tilde{W}$ is negative semidefinite. It implies $\tr(\tilde{W})<0$. However, we have
	\begin{eqnarray}
	\label{eq:c=2-3}
	\begin{split}
	\tr(\tilde{W})&=\tr\Big(\big((I_2\oplus 0_{n-2})\ox (I_n-(I_2\oplus 0_{n-2}))\big) W \big((I_2\oplus 0_{n-2})\ox (I_n-(I_2\oplus 0_{n-2}))\big)\Big)\\
	&=\sum_{i,j}\bra{i,j}\big((I_2\oplus 0_{n-2})\ox (I_n-(I_2\oplus 0_{n-2}))\big) W \big((I_2\oplus 0_{n-2})\ox \\&(I_n-(I_2\oplus 0_{n-2}))\big)\ket{i,j}\\
	&\geq 0.
	\end{split}
	\end{eqnarray}
	The last inequality follows from $W$ is an EW. Thus, we derive a contradiction. Finally, suppose $m,n>2$. Using the projector $(I_m-(I_2\oplus 0_{m-2}))\ox (I_n-(I_2\oplus 0_{n-2}))$ we similarly obtain that $\tilde{W}=(I_m-(I_2\oplus 0_{m-2}))\ox (I_n-(I_2\oplus 0_{n-2}))W(I_m-(I_2\oplus 0_{m-2}))\ox (I_n-(I_2\oplus 0_{m-2}))$ is negative semidefinite. It follows that $\tr(\tilde{W})<0$. However, for the same reason \eqref{eq:c=2-3} we also have $\tr(\tilde{W})\geq 0$. We derive a contradiction again. Therefore, such an EW which has exact two positive eigenvalues does not exist.
	
	This completes the proof.
\end{proof}

We give some observations on the relation between $v_-$ and $v_+$ on the inertia.

\begin{lemma}
\label{le:negative}
Suppose $W$ is an EW on $\bbC^m \ox \bbC^n$. Suppose the negative space of $W$ is spanned by $\{\ket{u_j}\}$ and the positive space is spanned by $\{\ket{v_j}\}$. If $\{\ket{u_j}\}$ is supported in $\bbC^{m_1}\ox \bbC^{n_1}$ and $\{\ket{v_j}\}$ is supported in $\bbC^{m_2}\ox \bbC^{n_2}$ and $m_1>m_2$ and $n_1>n_2$,  then we can choose a projector $P=((I_{m_1}-(I_{m_2}\oplus 0_{m_1-m_2}))\oplus 0_{m-m_1})\ox((I_{n_1}-(I_{n_2}\oplus 0_{n_1-n_2}))\oplus 0_{n-n_1})$ such that $V=PWP$ is negative semidefinite.
\end{lemma}

\section{Inertias of the partial transpose of the two-qutrit bipartite NPT states}
\label{sec:I_2}
	In this section we focus on the inertias of  the two-qutrit bipartite NPT states.  In Subsec. \ref{subsec:4_1},  we propose two observations on of the two-qutrit  NPT states in Lemmas \ref{le:13} and \ref{le:12} firstly.  For understanding the partial transpose of the two-qutrit bipartite NPT states in Example \ref{ex:11}. Then we discuss all cases in  the two-qutrit bipartite NPT state in detail in Theorem \ref{th:NPT2}. In Subsec. \ref{subsec:4_2}, we draw a conclusion about $\mathcal{N}_{3,3}$  in Theorem \ref{th:N_33}. Finally we propose Example \ref{ex:13arrays} to illustrate Theorem \ref{th:N_33}(i).
\subsection{Supporting lemmas}
\label{subsec:4_1}	
At first, we propose some observations on the two-qutrit  NPT states.

\begin{lemma}
	\label{le:12}
	Let $\rho$ be a two-qutrit NPT state. Then we have,
	
	(i) if $|0,0\rangle$ and $|0,1\rangle$ are in $\ker \rho^{\Gamma}$, then $\ln (\rho^{\Gamma})$ is one of the $13$ arrays in Theorem \ref{th:N_33}(i).	
	
	(ii) $(4,2,3)\notin \mathcal{N}_{3,3}$.
	
\end{lemma}
\begin{proof}
	(i) If $|0,0\rangle \in \ker \rho^{\Gamma}$, then we have $\rho^{\Gamma} |0,0\rangle =0$. Then we have that the first row and column of $\rho^{\Gamma}$ consist of zero entries. Similarly, if $|0,1\rangle \in \ker \rho^{\Gamma}$, then we have $\rho^{\Gamma} |0,1\rangle =0$. Then we have that the second row and column of $\rho^{\Gamma}$ consist of zero entries. 
	
	So $\rho^{\Gamma}$ can be written as 	
	$	\bma 
	0 & 0 & 0 & 0 & 0 & 0 & 0 & 0 & 0 \\
	0 & 0 & 0 & 0 & 0 & 0 & 0 & 0 & 0 \\
	0 & 0 & * & * & * & * & * & * & * \\
	0 & 0 & * & * & * & * & * & * & * \\
	0 & 0 & * & * & * & * & * & * & * \\
	0 & 0 & * & * & * & * & * & * & * \\
	0 & 0 & * & * & * & * & * & * & * \\
	0 & 0 & * & * & * & * & * & * & * \\
	0 & 0 & * & * & * & * & * & * & * 
	\ema.$
	Because $ |0,0\rangle ,|0,1\rangle \in \ker(\rho^{\Gamma})$, we have $ |0,0\rangle, |0,1\rangle \in \ker(\rho)$.
	Then $\rho$ can be written as	
	$	\bma 
	0 & 0 & 0 & 0 & 0 & 0 & 0 & 0 & 0 \\
	0 & 0 & 0 & 0 & 0 & 0 & 0 & 0 & 0 \\
	0 & 0 & a & 0 & 0 & b & 0 & 0 & c \\
	0 & 0 & 0 & * & * & * & * & * & * \\
	0 & 0 & 0 & * & * & * & * & * & * \\
	0 & 0 & b^* & * & * & d & * & * & * \\
	0 & 0 & 0 & * & * & * & * & * & * \\
	0 & 0 & 0 & * & * & * & * & * & * \\
	0 & 0 & c^* & * & * & * & * & * & * 
	\ema$.	
	We have $\bma a & b \\ b^* & d \ema \geq 0 $ by $\rho \geq 0$. If $a=0$, then $b=c=0$.
	
	If $a\neq 0$, then we have   
	$\rho=	\bma 
	0 & 0 & 0 & 0 & 0 & 0 & 0 & 0 & 0 \\
	0 & 0 & 0 & 0 & 0 & 0 & 0 & 0 & 0 \\
	0 & 0 & a & 0 & 0 & b & 0 & 0 & c \\
	0 & 0 & 0 & * & * & * & * & * & * \\
	0 & 0 & 0 & * & * & * & * & * & * \\
	0 & 0 & b^* & * & * & d & * & * & * \\
	0 & 0 & 0 & * & * & * & * & * & * \\
	0 & 0 & 0 & * & * & * & * & * & * \\
	0 & 0 & c^* & * & * & * & * & * & * 
	\ema \geq 0$. There exists a local invertible product operator $V=A\ox B$ such that $V\rho V^+ = \bma 0&0&0\\ 0&0&0\\ 0&0&a \ema \oplus \sigma$, where $\sigma$ is a $2\times 3$ state.
	
	We have known that $ \mathcal{N}_{2,3}=\{(1,2,3), (1,1,4),(1,0,5), (2,0,4)\}$ from \cite[Corollary $11(a)$]{2020Inertias}. Then one can verify that $\ln (\rho^{\Gamma})$ is one of the 13 arrays in Theorem \ref{th:N_33}(i).
	
	(ii) Suppose $\ln(\rho^{\Gamma})=(4,2,3)$. Then we have $\dim(\ker(\rho^{\Gamma}))=2$. Because the dimension of  non-positive eigen-space of $\rho$ is $6>(3-1)\times(3-1)+1$, we obtain that the non-positive eigen-space of $\rho$ has infinitely many product vectors from Lemma \ref{le:2}. From Lemma \ref{le:mn} we obtain that $\ker(\rho^{\Gamma})$ contains  infinitely many product vectors. We have $\ker(\rho^{\Gamma})=\mathrm{span}  \{\ket{a,b}, \ket{c,d}\}$. If  $\ket{a}$ and $\ket{c}$ are linearly dependent, then we can find an invertible operator $W=C\ox D$ such that $W\ket{a,b}=\ket{0,0}$ and $W\ket{c,d}=\ket{0,1}$. Then we have $\ker((W^{-1})^{\dagger}\rho W^{-1})=\mathrm{span} \{\ket{0,0},\ket{0,1}\}$. Let $\sigma =(W^{-1})^{\dagger}\rho W^{-1}$. We obtain that $\ln(\sigma^{\Gamma})=\ln(\rho^{\Gamma})$ by Sylvester theorem.  Then $\ket{0,0}$ and $\ket{0,1}$ are in $\ker(\sigma^{\Gamma})$. From (i) we derive a contradiction.
	So we have $(4,2,3)\notin \mathcal{N}_{3,3}$.
	
\end{proof}

\begin{lemma}
	\label{le:13}
	Let $\rho$ be a two-qutrit NPT state. Then we have,
	
	(i) if $\ket{0,0}$, $\ket{1,1}$ and $(\ket{0}+\ket{1})(\ket{0}+\ket{1})$ are in $\ker \rho^{\Gamma}$, then $\ln (\rho^{\Gamma})$ is one of the $13$ arrays in Theorem \ref{th:N_33}(i).
	
	(ii) if $\ket{0,0}, \ket{1,1}, \ket{2}(\ket{0}+\ket{1}), (a\ket{0}+b\ket{1}+\ket{2})(\ket{0}+f\ket{1})$ satisfying the conditions $a,b\neq 0$  and $f\neq 0,1$ are in $\ker \rho^{\Gamma}$, then $\ln (\rho^{\Gamma})$ is one of the $12$ arrays in theorem \ref{th:N_33}(i).
	
	(iii) $(3,3,3), (2,4,3)\notin  \mathcal{N}_{3,3}$.
\end{lemma}
\begin{proof}
	(i) If $\ket{0,0} \in \ker \rho^{\Gamma}$, then we have $\rho^{\Gamma} \ket{0,0} =0$. Then we have that the first row and column of $\rho^{\Gamma}$ consist of zero entries. Similarly, if $\ket{1,1} \in \ker \rho^{\Gamma}$, then we have $\rho^{\Gamma} \ket{1,1} =0$. Then we have that the fifth row and column of $\rho^{\Gamma}$ consist of zero entries. We have  $(\ket{0}+\ket{1})(\ket{0}+\ket{1}) \in \ker \rho^{\Gamma}$.
	
	So $\rho^{\Gamma}$ can be written as 	
	$	\bma 
	0 & 0 & 0 & 0 & 0 & 0 & 0 & 0 & 0 \\
	0 & a & * & -a & 0 & * & * & * & * \\
	0 & b & * & -b & 0 & * & * & * & * \\
	0 & c & * & -c & 0 & * & * & * & * \\
	0 & 0 & 0 & 0 & 0 & 0 & 0 & 0 & 0 \\
	0 & d & * & -d & 0 & * & * & * & * \\
	0 & e & * & -e & 0 & * & * & * & * \\	
	0 & f & * & -f & 0 & * & * & * & * \\
	0 & g & * & -g & 0 & * & * & * & * 
	\ema.$ Because $ \ket{0,0}, \ket{1,1},  (\ket{0}+\ket{1})(\ket{0}+\ket{1}) \in \ker(\rho^{\Gamma})$, we have $ \ket{0,0}, \ket{1,1},  (\ket{0}+\ket{1})(\ket{0}+\ket{1})  \in \ker(\rho)$.
	Then $\rho$ can be written as	
	$	\bma 
	0 & 0 & 0 & 0 & 0 & 0 & 0 & 0 & 0 \\
	0 & 0 & 0 & 0 & 0 & 0 & 0 & 0 & h^* \\
	0 & 0 & * & 0 & 0 & * & 0 & g & * \\
	0 & 0 & 0 & 0 & 0 & 0 & 0 & 0 & -h^* \\
	0 & 0 & 0 & 0 & 0 & 0 & 0 & 0 & 0 \\
	0 & 0 & * & 0 & 0 & * & -g & 0 & * \\
	0 & 0 & 0 & 0 & 0 & -g^* & * & * & * \\
	0 & 0 & g^* & 0 & 0 & 0 & * & * & * \\
	0 & h & * & -h & 0 & * & * & * & * 
	\ema.$
	
	Because $\rho$ is a  positive semi-definite matrix, then $\rho$ can be written as 
	$	\bma 
	0 & 0 & 0 & 0 & 0 & 0 & 0 & 0 & 0 \\
	0 & 0 & 0 & 0 & 0 & 0 & 0 & 0 & 0 \\
	0 & 0 & * & 0 & 0 & * & 0 & g & * \\
	0 & 0 & 0 & 0 & 0 & 0 & 0 & 0 & 0 \\
	0 & 0 & 0 & 0 & 0 & 0 & 0 & 0 & 0 \\
	0 & 0 & * & 0 & 0 & * & -g & 0 & * \\
	0 & 0 & 0 & 0 & 0 & -g^* & * & * & * \\
	0 & 0 & g^* & 0 & 0 & 0 & * & * & * \\
	0 & 0 & * & 0 & 0 & * & * & * & * 
	\ema.$ 
	
	Then we obtain that $\ket{0,0}$ and $\ket{0,1}$ are in the  $\ker (\rho^{\Gamma})$. We obtain that $\ln (\rho^{\Gamma})$ is one of the $13$ arrays in Theorem \ref{th:N_33}(i) from Lemma \ref{le:12}(i). 
	
	(ii) If $\ket{0,0} \in \ker \rho^{\Gamma}$, then we have $\rho^{\Gamma} \ket{0,0} =0$. Then we have that the first row and column of $\rho^{\Gamma}$ consist of zero entries. Similarly, if $\ket{1,1} \in \ker \rho^{\Gamma}$, then we have $\rho^{\Gamma} \ket{1,1} =0$. Then we have that the fifth row and column of $\rho^{\Gamma}$ consist of zero entries. We have  $\ket{2}(\ket{0}+\ket{1}) \in \ker \rho^{\Gamma}$. Then we have that $\rho^{\Gamma}\ket{2}(\ket{0}+\ket{1})=0$.
	
	So $\rho^{\Gamma}$ can be written as 	
	$	\bma 
	0 & 0 & 0 & 0 & 0 & 0 & 0 & 0 & 0 \\
	0 & * & * & * & 0 & * & l & -l & * \\
	0 & * & * & * & * & * & m & -m & * \\
	0 & * & * & * & 0 & * & n & -n & * \\
	0 & 0 & 0 & 0 & 0 & 0 & 0 & 0 & 0 \\
	0 & * & * & * & 0 & * & o & -o & * \\
	0 & * & * & * & 0 & * & p & -p & * \\
	0 & * & * & * & 0 & * & q & -q & * \\
	0 & * & * & * & 0 & * & r & -r& * 
	\ema.$
	Because $ \ket{0,0}, \ket{1,1},  \ket{2}(\ket{0}+\ket{1}), (a\ket{0}+b\ket{1}+\ket{2})(\ket{0}+f\ket{1}) \in \ker(\rho^{\Gamma})$, we have $ \ket{0,0}, \ket{1,1}, \ket{2}(\ket{0}+\ket{1}), (a\ket{0}+b\ket{1}+\ket{2})(\ket{0}+f\ket{1})  \in \ker(\rho)$. Then we have that $\rho \ket{0,0} =0$, $\rho \ket{1,1} =0$, $\rho \ket{2}(\ket{0}+\ket{1})=0$ and  $\rho (a\ket{0}+b\ket{1}+\ket{2})(\ket{0}+f\ket{1})=0$. 
	Then $\rho$ can be written as 
	$	\bma 
	0 & 0 & 0 & 0 & 0 & 0 & 0 & 0 & 0 \\
	0 & 0 & * & 0 & 0 & 0 & 0 & 0 & * \\
	0 & 0 & * & 0 & 0 & * & 0 & 0 & * \\
	0 & 0 & 0 & 0 & 0 & * & 0 & 0 & * \\
	0 & 0 & 0 & 0 & 0 & 0 & 0 & 0 & 0 \\
	0 & 0 & * & 0 & 0 & * & 0 & 0 & * \\
	0 & 0 & 0 & 0 & 0 & * & 0 & 0 & * \\
	0 & 0 & * & 0 & 0 & * & 0 & 0 & * \\
	0 & 0 & * & 0 & 0 & * & 0 & 0 & * 
	\ema.$ 
	Because $\rho$ is a  positive semi-definite matrix and $\rho$ is a hermitian matrix, then $\rho$ can be written as 
	$	\bma 
	0 & 0 & 0 & 0 & 0 & 0 & 0 & 0 & 0 \\
	0 & 0 & 0 & 0 & 0 & 0 & 0 & 0 & 0 \\
	0 & 0 & * & 0 & 0 & * & 0 & 0 & * \\
	0 & 0 & 0 & 0 & 0 & 0 & 0 & 0 & 0 \\
	0 & 0 & 0 & 0 & 0 & 0 & 0 & 0 & 0 \\
	0 & 0 & * & 0 & 0 & * & 0 & 0 & * \\
	0 & 0 & 0 & 0 & 0 & 0 & 0 & 0 & 0 \\
	0 & 0 & 0 & 0 & 0 & 0 & 0 & 0 & 0 \\
	0 & 0 & * & 0 & 0 & * & 0 & 0 & * 
	\ema.$ 
	
	Then we obtain that $\ket{0,0}$ and $\ket{0,1}$ are in the  $\ker (\rho^{\Gamma})$. We obtain that $\ln (\rho^{\Gamma})$ is one of the $13$ arrays in Theorem \ref{th:N_33}(i) from Lemma \ref{le:12}(i).
	
	(iii) If $\ln (\rho^{\Gamma})= (3,3,3)$ or $(2,4,3)$, we obtain that the non-positive eigen-space of $\rho$ has infinitely many product vectors form Lemma \ref{le:2}. Then we obtain that the zero eigen-space of $\rho$ has infinitely many product vectors from Lemma \ref{le:mn}(i). If $\ln (\rho^{\Gamma})= (3,3,3)$, then the zero eigen-space of $\rho$ is up to  
	SLOCC equivalence  spanned by $\{\ket{0,0},\ket{0,1}\}$ or  $\{\ket{0,0},\ket{1,1}, (\ket{0}+\ket{1})(\ket{0}+\ket{1})\}$ by Lemma \ref{le:9}(ii).  If $\ln (\rho^{\Gamma})= (2,4,3)$, then the zero eigen-space of $\rho$ is up to  
	SLOCC equivalence  spanned by $\{\ket{0,0},\ket{0,1}\}$ or $\{\ket{0,0},\ket{1,1}, (\ket{0}+\ket{1})(\ket{0}+\ket{1})\}$ or $\{\ket{0,0}, \ket{1,1}, \ket{2}(\ket{0}+\ket{1}), (a\ket{0}+b\ket{1}+\ket{2})(\ket{0}+f\ket{1})\}$ satisfying the conditions $a,b\neq 0$  and $f\neq 0,1$ by Lemma \ref{le:9}(iii). Then in these two cases we obtain that $\ket{0,0}$ and $\ket{0,1}$ are in the  $\ker (\rho^{\Gamma})$. $\ln (\rho^{\Gamma})$ is one of the 12 arrays in Theorem \ref{th:N_33}(i) from (i) and (ii). Then we derive a contradiction. So we have  $(3,3,3), (2,4,3)\notin  \mathcal{N}_{3,3}$.
\end{proof}
We propose an example to understand the two-qutrit bipartite NPT states in Example \ref{ex:11}. Then we discussed all cases in the two-qutrit bipartite NPT states in detail in Theorem \ref{th:NPT2}.
\begin{example}
	\label{ex:11}
	Suppose $\rho =(\ket{0,0}+\ket{1,1}+\ket{2,2})(\bra{0,0}+\bra{1,1}+\bra{2,2})+(\ket{0}(a\ket{0}+b\ket{1}))(\bra{0}(a^*\bra{0}+b^*\bra{1}))$ is a two-qutrit bipartite NPT state. 
	
	Then $\rho$ can be written as	
	$	\bma 
	1+\left| a \right|^2 & ab^* & 0 & 0 & 1 & 0 & 0 & 0 & 1 \\
	a^*b & \left| b \right|^2 & 0 & 0 & 0 & 0 & 0 & 0 & 0 \\
	0 & 0 & 0 & 0 & 0 & 0 & 0 & 0 & 0 \\
	0 & 0 & 0 & 0 & 0 & 0 & 0 & 0 & 0 \\
	1 & 0 & 0 & 0 & 1 & 0 & 0 & 0 & 1 \\
	0 & 0 & 0 & 0 & 0 & 0 & 0 & 0 & 0 \\
	0 & 0 & 0 & 0 & 0 & 0 & 0 & 0 & 0 \\	
	0 & 0 & 0 & 0 & 0 & 0 & 0 & 0 & 0 \\	
	1 & 0 & 0 & 0 & 1 & 0 & 0 & 0 & 1 \\
	\ema$.
	Then $\rho^\G$ can be written as	
	$	\bma 
	1+\left| a \right|^2 & ab^* & 0 & 0 & 0 & 0 & 0 & 0 & 0 \\
	a^*b & \left| b \right|^2 & 0 & 1 & 0 & 0 & 0 & 0 & 0 \\
	0 & 0 & 0 & 0 & 0 & 0 & 1 & 0 & 0 \\
	0 & 1 & 0 & 0 & 0 & 0 & 0 & 0 & 0 \\
	0 & 0 & 0 & 0 & 1 & 0 & 0 & 0 & 0 \\
	0 & 0 & 0 & 0 & 0 & 0 & 0 & 1 & 0 \\
	0 & 0 & 1 & 0 & 0 & 0 & 0 & 0 & 0 \\	
	0 & 0 & 0 & 0 & 0 & 1 & 0 & 0 & 0 \\	
	0 & 0 & 0 & 0 & 0 & 0 & 0 & 0 & 1 \\
	\ema$.
	Then we obtain that $\det \r^\G=-1-\left| a \right|^2$. We get that $\r^\G$ has six constant eigenvalues, $-1,-1,1,1,1,1$. The other three eigenvalues of $\r^\G$ are the three roots of $x^3+(-1-\left| a \right|^2-\left| b \right|^2)x^2+(-1+\left| b \right|^2)x+1+\left| a \right|^2$. Suppose these three roots are $x_1>x_2>x_3$ respectively. We obtain that $x_1+x_2+x_3=1-\left| b \right|^2$ and $x_1x_2x_3=-1-\left| a \right|^2$. Then we obtain that three eigenvalues are $1$, $\frac{1}{2}(\left| a \right|^2+\left| b \right|^2-\sqrt{4+4\left| a \right|^2+\left| a \right|^4+2\left| a \right|^2\left| b \right|^2+\left| b \right|^4})$ and $\frac{1}{2}(\left| a \right|^2+\left| b \right|^2+\sqrt{4+4\left| a \right|^2+\left| a \right|^4+2\left| a \right|^2\left| b \right|^2+\left| b \right|^4})$. We get that $x_1,x_2>0$ and $x_3<0$. So nine eigenvalues of $\r^\G$ are $-1,-1,x_3,1,1,1,1,x_1,x_2$. The interia $\ln \r^\G=(3,0,6)$.
	
\end{example}

In this example, we investigate the interia of a two-qutrit bipartite NPT state $\r=\proj{\a}+\proj{\b}$ of rank two. We define the case of SR$(\ket{\a})=m$ and SR$(\b)=n$ as a binary array $(m,n)$. We get that $\r$ has six cases, $(1,1)$, $(2,1)$, $(2,2)$, $(3,1)$, $(3,2)$, $(3,3)$. In the case of $(1,1)$, $\r$ is a separable state. We need to get rid of $(1,1)$.  On the other hand, the case $(3,3)$ is equivalent to the one of the remaining four cases. So it suffices to study four cases of two-qutrit bipartite NPT states of rank two as follows.

\begin{theorem}
	\label{th:NPT2}
	Suppose $\r=\proj{\a}+\proj{\b}$ is a two-qutrit bipartite NPT state. 
	
	(i) If SR$(\ket{\a})=2$ and SR$(\ket{\b})=1$, then we have $\ln \r^\G=(1,4,4)$.
	
	(ii) If SR$(\ket{\a})=2$ and SR$(\ket{\b})=2$, then we have $\ln \r^\G= (2,2,5)$.	
	
	(iii) If SR$(\ket{\a})=3$ and SR$(\ket{\b})=1$, then we have $\ln \r^\G=(3,0,6)$.	
	
	(iv) If SR$(\ket{\a})=3$ and SR$(\ket{\b})=2$, then we have $\ln \r^\G=(2,2,5),(2,1,6),(2,0,7),(3,0,6),(3,1,5)$.
	
\end{theorem}

\begin{proof}
	Suppose $\r=\proj{\a}+\proj{\b}$ is a two-qutrit bipartite NPT state,
	\begin{eqnarray}
		\label{eq:NPT}
		\r=\proj{\a}+\proj{\b}
	\end{eqnarray}.
	
	(i) If SR$(\ket{\a})=2$ in Eq\eqref{eq:NPT}, then we have that $\ket{\a}$ is up to SLOCC equivalent to $\ket{0,0}+\ket{1,1}$.  If SR$(\ket{\b})=1$, then we have that $\ket{\b}$ is up to LOCC equivalent to $\ket{2,2}$. Then we have $\r^\G=\bma 
	1 & 0 & 0 & 0 & 0 & 0 & 0 & 0 & 0 \\
	0 & 0 & 0 & 1 & 0 & 0 & 0 & 0 & 0 \\
	0 & 0 & 0 & 0 & 0 & 0 & 0 & 0 & 0 \\
	0 & 1 & 0 & 0 & 0 & 0 & 0 & 0 & 0 \\
	0 & 0 & 0 & 0 & 1 & 0 & 0 & 0 & 0 \\
	0 & 0 & 0 & 0 & 0 & 0 & 0 & 0 & 0 \\
	0 & 0 & 0 & 0 & 0 & 0 & 0 & 0 & 0 \\	
	0 & 0 & 0 & 0 & 0 & 0 & 0 & 0 & 0 \\	
	0 & 0 & 0 & 0 & 0 & 0 & 0 & 0 & 1 \\
	\ema$. Then we get that nine eigenvalues are $-1,0,0,0,0,1,1,1,1$. Then $\ln \r^\G=(1,4,4) $. In conclusion, we obtain that if SR$(\ket{\a})=2$ and SR$(\ket{\b})=1$, then we have $\ln \r^\G=(1,4,4)$. 
	
	(ii) If SR$(\ket{\a})=2$ in Eq\eqref{eq:NPT}, then we have that $\ket{\a}$ is up to SLOCC equivalent to $\ket{0,0}+\ket{1,1}$.  If SR$(\ket{\b})=2$, then we have that $\ket{\b}$ cotains $\ket{2,y_2}$ and $(a\ket{0}+b\ket{1})(\ket{y_2})$. Because $(a\ket{0}+b\ket{1})(\ket{y_0})$ is equivalent to $\ket{0,y_0}$, we have that $\ket{\b}$ is up to SLOCC equivalent to $\ket{0,y_0}+\ket{2,y_2}$. We get that at least one of $\ket{y_0}$ and $\ket{y_2}$ contains $\ket{2}$. So $\ket{\b}$ has two cases. One case is that $\ket{\b}$ is up to SLOCC equivalent to $\ket{0}(a\ket{0}+b\ket{1})+\ket{2,2}$, the other case is  that $\ket{\b}$ is up to SLOCC equivalent to $\ket{0,2}+\ket{2}(a\ket{0}+b\ket{1})$.

	(ii.a) Suppose $\ket{\b}=\ket{0}(a\ket{0}+b\ket{1})+\ket{2,2}$ in Eq\eqref{eq:NPT}.
	Then we have $\r^\G=\bma
	1+\left| a \right|^2 & ab^* & 0 & 0 & 0 & 0 & 0 & 0 & 0 \\
	ba^* & \left| b \right|^2 & 0 & 1 & 0 & 0 & 0 & 0 & 0 \\
	0 & 0 & 0 & 0 & 0 & 0 & a^* & b^* & 0 \\
	0 & 1 & 0 & 0 & 0 & 0 & 0 & 0 & 0 \\
	0 & 0 & 0 & 0 & 1 & 0 & 0 & 0 & 0 \\
	0 & 0 & 0 & 0 & 0 & 0 & 0 & 0 & 0 \\
	0 & 0 & a & 0 & 0 & 0 & 0 & 0 & 0 \\	
	0 & 0 & b & 0 & 0 & 0 & 0 & 0 & 0 \\	
	0 & 0 & 0 & 0 & 0 & 0 & 0 & 0 & 1 \\
	\ema$. Then we get that nine eigenvalues are $-(a^2+b^2)^{-\frac{1}{2}},\frac{1}{2}(a^2+b^2-(4+4a^2+a^4+2a^2b^2+b^4)^{-\frac{1}{2}},0,0,1,1,1,(a^2+b^2)^{-\frac{1}{2}},\frac{1}{2}(a^2+b^2-(4+4a^2+a^4+2a^2b^2+b^4)^{-\frac{1}{2}}$. Then $\ln \r^\G=(2,2,5)$.

	(ii.b) Suppose $\ket{\b}=\ket{0,2}+\ket{2}(a\ket{0}+b\ket{1})$ in Eq\eqref{eq:NPT}.
	
	Then we have $\r^\G=\bma 
	1 & 0 & 0 & 0 & 0 & 0 & 0 & 0 & a \\
	0 & 0 & 0 & 1 & 0 & 0 & 0 & 0 & b \\
	0 & 0 & 1 & 0 & 0 & 0 & 0 & 0 & 0 \\
	0 & 1 & 0 & 0 & 0 & 0 & 0 & 0 & 0 \\
	0 & 0 & 0 & 0 & 1 & 0 & 0 & 0 & 0 \\
	0 & 0 & 0 & 0 & 0 & 0 & 0 & 0 & 0 \\
	0 & 0 & 0 & 0 & 0 & 0 & \left| a \right|^2 & ab^* & 0 \\	
	0 & 0 & 0 & 0 & 0 & 0 & ba^* & \left| b \right|^2 & 0 \\	
	a^* & b^* & 0 & 0 & 0 & 0 & 0 & 0 & 0 \\
	\ema$. Then we get that six eigenvalues are $0,0,1,1,1,a^2+b^2$. The other three eigenvalues are the roots of $x^3+(-1-a^2-b^2)x-a^2$. These three roots are two negative roots and one positive root. Then $\ln \r^\G=(2,2,5)$. In conclusion, we obtain that if SR$(\ket{\a})=2$ and SR$(\ket{\b})=2$, then we have $\ln \r^\G=(2,2,5)$. 
	
	(iii) If SR$(\ket{\a})=3$ in Eq\eqref{eq:NPT}, then we have that $\ket{\a}$ is up to SLOCC equivalent to $\ket{0,0}+\ket{1,1}+\ket{2,2}$.  If SR$(\ket{\b})=1$, then we have that $\ket{\b}$ is up to SLOCC equivalent to $\ket{0}(a\ket{0}+b\ket{1})$.
	From Example \ref{ex:11}, we get that $\ln \r^\G=(3,0,6)$.
	
	(iv) If SR$(\ket{\a})=3$ in Eq\eqref{eq:NPT}, then we have that $\ket{\a}$ is up to SLOCC equivalent to $\ket{0,0}+\ket{1,1}+\ket{2,2}$.  If SR$(\ket{\b})=2$, then we have that $\ket{\b}$ is up to SLOCC equivalent to $\ket{0}(a\ket{0}+b\ket{1}+c\ket{2})+\ket{1}(d\ket{0}+e\ket{1})$. Suppose $M$ is a two-qutrit invertible matrix. There exists a two-qutrit invertible matrix $M$, such that $(M(M^T)^{-1}\otimes I_3)(\ket{0,0}+\ket{1,1}+\ket{2,2})=(\ket{0,0}+\ket{1,1}+\ket{2,2})$ and $M(M^T)^{-1}(\ket{0}(a\ket{0}+b\ket{1}+c\ket{2})+\ket{1}(d\ket{0}+e\ket{1}))=\ket{0}(a'\ket{0}+b'\ket{1}+c'\ket{2})+e'\ket{1,1}$ or $\ket{0}(a'\ket{0}+b'\ket{1}+c'\ket{2})+d'\ket{1,0}$. Then $\ket{\b}$ has two cases, $\ket{0}(a\ket{0}+b\ket{1}+c\ket{2})+e\ket{1,1}$ and $\ket{0}(a\ket{0}+b\ket{1}+c\ket{2})+d\ket{1,0}$. The case  $\ket{0}(a\ket{0}+b\ket{1}+c\ket{2})+d\ket{1,0}$ is SLOCC equivalent to the case $\ket{0}(a\ket{0}+b\ket{1}+c\ket{2})+\ket{1,0}$. There exists a two-qutrit invertible matrix $M_1$, such that $(M_1(M_1^T)^{-1}\otimes I_3)(\ket{0,0}+\ket{1,1}+\ket{2,2})=(\ket{0,0}+\ket{1,1}+\ket{2,2})$ and $(M_1(M_1^T)^{-1}\otimes I_3)(\ket{0}(a\ket{0}+b\ket{1}+c\ket{2})+d\ket{1,0})=\ket{0}(a''\ket{0}+b''\ket{1})+d''\ket{1,0}$ or $\ket{0}(a''\ket{0}+c''\ket{2})+d''\ket{1,0}$. The case  $\ket{0}(a\ket{0}+b\ket{1}+c\ket{2})+\ket{1,0}$ is divided into two cases $\ket{0}(a\ket{0}+b\ket{1})+\ket{1,0}$ and  $\ket{0}(a\ket{0}+c\ket{2})+\ket{1,0}$. Similarly, the case  $\ket{0}(a\ket{0}+b\ket{1}+c\ket{2})+e\ket{1,1}$ is divided into two cases $\ket{0}(a\ket{0}+b\ket{1})+e\ket{1,1}$ and  $\ket{0}(b\ket{1}+c\ket{2})+e\ket{1,1}$.  If SR$(\ket{\b})=2$, then $\ket{\b}$ has four cases $\ket{0}(a\ket{0}+b\ket{1})+\ket{1,0}$, $\ket{0}(a\ket{0}+c\ket{2})+\ket{1,0}$, $\ket{0}(a\ket{0}+b\ket{1})+e\ket{1,1}$, $\ket{0}(b\ket{1}+c\ket{2})+e\ket{1,1}$.
	
	(iv.a) Suppose $\ket{\b}=\ket{0}(a\ket{0}+b\ket{1})+\ket{1,0}$ in Eq\eqref{eq:NPT}. Then we have $\r^\G=\bma
	1+\left| a \right|^2 & ab^* & 0 & a^* & b^* & 0 & 0 & 0 & 0 \\
	ba^* & \left| b \right|^2 & 0 & 1 & 0 & 0 & 0 & 0 & 0 \\
	0 & 0 & 0 & 0 & 0 & 0 & 1 & 0 & 0 \\
	a & 1 & 0 & 1 & 0 & 0 & 0 & 0 & 0 \\
	b & 0 & 0 & 0 & 1 & 0 & 0 & 0 & 0 \\
	0 & 0 & 0 & 0 & 0 & 0 & 0 & 1 & 0 \\
	0 & 0 & 1 & 0 & 0 & 0 & 0 & 0 & 0 \\	
	0 & 0 & 0 & 0 & 0 & 1 & 0 & 0 & 0 \\	
	0 & 0 & 0 & 0 & 0 & 0 & 0 & 0 & 1 \\
	\ema$. We know that $A\r^\G A^\dagger$ and $\r^\G$ have the same inertia when $A$ is an elementary matrix. We obtain that $S=B\r^\G B^\dagger=\bma
	1 & 0 & 0 & 0 & 0 & 0 & 0 & 0 & 0 \\
	0 & \left| b \right|^2-\left| ab^*-a^* \right|^2-1 & 0 & 0 & 0 & 0 & 0 & 0 & 0 \\
	0 & 0 & 0 & 0 & 0 & 0 & 1 & 0 & 0 \\
	0 & 0 & 0 & 1 & 0 & 0 & 0 & 0 & 0 \\
	0 & 0 & 0 & 0 & 1 & 0 & 0 & 0 & 0 \\
	0 & 0 & 0 & 0 & 0 & 0 & 0 & 1 & 0 \\
	0 & 0 & 1 & 0 & 0 & 0 & 0 & 0 & 0 \\	
	0 & 0 & 0 & 0 & 0 & 1 & 0 & 0 & 0 \\	
	0 & 0 & 0 & 0 & 0 & 0 & 0 & 0 & 1 \\
	\ema$ has the same inertia as $\r^\G$, where $B$ is a product of some elementary matrices.  Then we get that the six eigenvalues of $\r^\G$ are $-1,-1,1,1,1,1,1,1, \left| b \right|^2-\left| ab^*-a^* \right|^2-1$. $\left| b \right|^2-\left| ab^*-a^* \right|^2-1$ can be a positive or zero or negative number. Then $\ln \r^\G=(3,0,6) $ when $\left| b \right|^2-\left| ab^*-a^* \right|^2-1<0$, $\ln \r^\G=(2,1,6) $ when $\left| b \right|^2-\left| ab^*-a^* \right|^2-1=0$, $\ln \r^\G=(2,0,7) $ when $\left| b \right|^2-\left| ab^*-a^* \right|^2-1>0$.
	
	(iv.b) Suppose $\ket{\b}=\ket{0}(a\ket{0}+c\ket{2})+\ket{1,0}$ in Eq\eqref{eq:NPT}. If $c=0$, then SR$(\ket{\b})=1$ and we reduce to (iii). Then we have $c\neq 0$ and $\bma
	1 & 0 & 0 \\
	0 & 1 & 0 \\
	0 & 0 & c \\
	\ema \otimes \bma
	1 & 0 & 0 \\
	0 & 1 & 0 \\
	0 & 0 & c^{-1} \\
	\ema \ket{\b}=\ket{0}(a\ket{0}+\ket{2})+\ket{1,0}$.  So we have $\ket{\b}$ is equivalent to $\ket{0}(a\ket{0}+\ket{2})+\ket{1,0}$. Then we have $\r^\G=\bma 
	1+\left| a \right|^2 & 0 & a & a^* & 0 & 1 & 0 & 0 & 0 \\
	0 & 0 & 0 & 1 & 0 & 0 & 0 & 0 & 0 \\
	a^* & 0 & 1 & 0 & 0 & 0 & 1 & 0 & 0 \\
	a & 1 & 0 & 1 & 0 & 0 & 0 & 0 & 0 \\
	0 & 0 & 0 & 0 & 1 & 0 & 0 & 0 & 0 \\
	1 & 0 & 0 & 0 & 0 & 0 & 0 & 1 & 0 \\
	0 & 0 & 1 & 0 & 0 & 0 & 0 & 0 & 0 \\	
	0 & 0 & 0 & 0 & 0 & 1 & 0 & 0 & 0 \\	
	0 & 0 & 0 & 0 & 0 & 0 & 0 & 0 & 1 \\
	\ema$. We have known that $A\r^\G A^\dagger$ and $\r^\G$ have the same inertia when $A$ is an elementary matrix. We obtain that $S=B\r^\G B^\dagger=\bma
	1+\left| a \right|^2 & 0 & 0 & 0 & 0 & 0 & 0 & 0 & 0 \\
	0 & 0 & 0 & 1 & 0 & 0 & 0 & 0 & 0 \\
	0 & 0 & 1 & 0 & 0 & 0 & 1 & 0 & 0 \\
	0 & 1 & 0 & 0 & 0 & 0 & 0 & 0 & 0 \\
	0 & 0 & 0 & 0 & 1 & 0 & 0 & 0 & 0 \\
	0 & 0 & 0 & 0 & 0 & 0 & 0 & 1 & 0 \\
	0 & 0 & 1 & 0 & 0 & 0 & 0 & 0 & 0 \\	
	0 & 0 & 0 & 0 & 0 & 1 & 0 & 0 & 0 \\	
	0 & 0 & 0 & 0 & 0 & 0 & 0 & 0 & 1 \\
	\ema$ has the same inertia as $\r^\G$, where $B$ is a product of some elementary matrices. Then we get that nine enigenvalues are $-1,-1,-1,1,1,1,1,1,1+\left| a \right|^2$. Then $\ln \r^\G=\ln S =(3,0,6)$.
	
	(iv.c) Suppose $\ket{\b}=\ket{0}(a\ket{0}+b\ket{1})+e\ket{1,1}$ in Eq\eqref{eq:NPT}. Then we have $\r^\G=\bma
	1+\left| a \right|^2 & ab^* & 0 & 0 & 0 & 0 & 0 & 0 & 0 \\
	ba^* & \left| b \right|^2 & 0 & 1+ea^* & eb^* & 0 & 0 & 0 & 0 \\
	0 & 0 & 0 & 0 & 0 & 0 & 1 & 0 & 0 \\
	0 & 1+ae^* & 0 & 0 & 0 & 0 & 0 & 0 & 0 \\
	0 & be^* & 0 & 0 & 1+\left| e \right|^2 & 0 & 0 & 0 & 0 \\
	0 & 0 & 0 & 0 & 0 & 0 & 0 & 1 & 0 \\
	0 & 0 & 1 & 0 & 0 & 0 & 0 & 0 & 0 \\	
	0 & 0 & 0 & 0 & 0 & 1 & 0 & 0 & 0 \\	
	0 & 0 & 0 & 0 & 0 & 0 & 0 & 0 & 1 \\
	\ema$. Suppose $1+ae^*\neq 0$.  Similar to (iv.a), we get that  $S=B\r^\G B^\dagger=\bma
	1+\left| a \right|^2 & 0 & 0 & 0 & 0 & 0 & 0 & 0 & 0 \\
	0 & 0 & 0 & 1+ea^* & 0 & 0 & 0 & 0 & 0 \\
	0 & 0 & 0 & 0 & 0 & 0 & 1 & 0 & 0 \\
	0 & 1+ae^* & 0 & 0 & 0 & 0 & 0 & 0 & 0 \\
	0 & 0 & 0 & 0 & 1+\left| e \right|^2 & 0 & 0 & 0 & 0 \\
	0 & 0 & 0 & 0 & 0 & 0 & 0 & 1 & 0 \\
	0 & 0 & 1 & 0 & 0 & 0 & 0 & 0 & 0 \\	
	0 & 0 & 0 & 0 & 0 & 1 & 0 & 0 & 0 \\	
	0 & 0 & 0 & 0 & 0 & 0 & 0 & 0 & 1 \\
	\ema$ has the same inertia as $\r^\G$, where $B$ is a product of some elementary matrices.  Then we get that nine enigenvalues are $-1,-1,-\sqrt{1+ae^*}\sqrt{1+ea^*},1,1,1,1+a^2,1+e^2,\sqrt{1+ae^*}\sqrt{1+ea^*}$. Then $\ln (\r^\G)=(3,0,6)$.
	If $1+ae^*=0$, we get that  $S'=B\r^\G B^\dagger=\bma
	1+\left| a \right|^2 & 0 & 0 & 0 & 0 & 0 & 0 & 0 & 0 \\
	0 & \left| b \right|^2-\frac{\left| be \right|^2}{1+\left| e \right|^2}-\frac{\left| ab \right|^2}{1+\left| a \right|^2} & 0 & 0 & 0 & 0 & 0 & 0 & 0 \\
	0 & 0 & 0 & 0 & 0 & 0 & 1 & 0 & 0 \\
	0 & 0 & 0 & 0 & 0 & 0 & 0 & 0 & 0 \\
	0 & 0 & 0 & 0 & 1+\left| e \right|^2 & 0 & 0 & 0 & 0 \\
	0 & 0 & 0 & 0 & 0 & 0 & 0 & 1 & 0 \\
	0 & 0 & 1 & 0 & 0 & 0 & 0 & 0 & 0 \\	
	0 & 0 & 0 & 0 & 0 & 1 & 0 & 0 & 0 \\	
	0 & 0 & 0 & 0 & 0 & 0 & 0 & 0 & 1 \\
	\ema$ has the same inertia as $\r^\G$, where $B$ is a product of some elementary matrices. Then we get that nine enigenvalues are $-1,-1,0,1,1,1,1+a^2,1+e^2,\left| b \right|^2-\frac{\left| be \right|^2}{1+\left| e \right|^2}-\frac{\left| ab \right|^2}{1+\left| a \right|^2}$. $\left| b \right|^2-\frac{\left| be \right|^2}{1+\left| e \right|^2}-\frac{\left| ab \right|^2}{1+\left| a \right|^2}$ can be a positive or zero or negative number. Then $\ln (\r^\G)=(3,1,5),(2,2,5),(2,1,6)$.
	In conclusion, we obtain that $\ln (\r^\G)=(3,0,6),(3,1,5),(2,2,5),(2,1,6)$.

	(iv.d) Suppose $\ket{\b}=\ket{0}(b\ket{1}+c\ket{2})+e\ket{1,1}$ in Eq\eqref{eq:NPT}. If $c=0$, then SR$(\ket{\b})=1$ and we reduce to (iii).  Then we have $c\neq 0$ and $\bma
	1 & 0 & 0 \\
	0 & 1 & 0 \\
	0 & 0 & c \\
	\ema \otimes \bma
	1 & 0 & 0 \\
	0 & 1 & 0 \\
	0 & 0 & c^{-1} \\
	\ema \ket{\b}=\ket{0}(a\ket{0}+\ket{2})+\ket{1,0}$. So $\ket{\b}$ is equivalent to $\ket{0}(b\ket{1}+\ket{2})+e\ket{1,1}$. Then we have $\r^\G=\bma
	1 & 0 & 0 & 0 & 0 & 0 & 0 & 0 & 0 \\
	0 & \left| b \right|^2 & b & 1 & eb^* & e & 0 & 0 & 0 \\
	0 & b^* & 1 & 0 & 0 & 0 & 1 & 0 & 0 \\
	0 & 1 & 0 & 0 & 0 & 0 & 0 & 0 & 0 \\
	0 & be^* & 0 & 0 & 1+\left| e \right|^2 & 0 & 0 & 0 & 0 \\
	0 & e^* & 0 & 0 & 0 & 0 & 0 & 1 & 0 \\
	0 & 0 & 1 & 0 & 0 & 0 & 0 & 0 & 0 \\	
	0 & 0 & 0 & 0 & 0 & 1 & 0 & 0 & 0 \\	
	0 & 0 & 0 & 0 & 0 & 0 & 0 & 0 & 1 \\
	\ema$. Similar to (iv.b), we get that  $S=B\r^\G B^\dagger=\bma
	1 & 0 & 0 & 0 & 0 & 0 & 0 & 0 & 0 \\
	0 & 0 & 0 & 1 & 0 & 0 & 0 & 0 & 0 \\
	0 & 0 & 0 & 0 & 0 & 0 & 1 & 0 & 0 \\
	0 & 1 & 0 & 0 & 0 & 0 & 0 & 0 & 0 \\
	0 & 0 & 0 & 0 & 1+\left| e \right|^2 & 0 & 0 & 0 & 0 \\
	0 & 0 & 0 & 0 & 0 & 0 & 0 & 1 & 0 \\
	0 & 0 & 1 & 0 & 0 & 0 & 0 & 0 & 0 \\	
	0 & 0 & 0 & 0 & 0 & 1 & 0 & 0 & 0 \\	
	0 & 0 & 0 & 0 & 0 & 0 & 0 & 0 & 1 \\
	\ema$ has the same inertia as $\r^\G$, where $B$ is a product of some elementary matrices.  Then we get that nine enigenvalues are $-1,-1,-1,1,1,1,1,1,1+\left| e \right|^2$. Then $\ln \r^\G=\ln S =(3,0,6)$.

	In conlusion, if SR$(\ket{\a})=3$ and SR$(\ket{\b})=2$, then we have $\ln \r^\G=(2,2,5),(2,1,6),(2,0,7),(3,0,6),(3,1,5)$.
	
\end{proof}
\subsection{Main conclusion}
\label{subsec:4_2}	
We have proposed some observations about two-qutrit bipartite NPT states  in Subsec. \ref{subsec:4_1}. Then we investigate the set $\mathcal{N}_{3,3}$. We have the following observations.
\begin{theorem}
	\label{th:N_33}
	(i) $\mathcal{N}_{3,3}$ has the following $13$ arrays,
	$(1,0,8)$, $(1,1,7)$, $(1,2,6)$, $(1,3,5)$, $(1,4,4)$, $(1,5,3)$, $(2,0,7)$, $(2,1,6)$, $(2,2,5)$, $(2,3,4)$, $(3,0,6)$, $(3,1,5)$, $(4,0,5)$.
	
	(ii) $\mathcal{N}_{3,3}$  does not have the following three arrays,
	$(2,4,3)$, $(3,3,3)$, $(4,2,3)$.
	
	(iii) $\mathcal{N}_{3,3}$ might have the following two arrays,
	$(3,2,4)$, $(4,1,4)$.
	
	(iv) $13\leq |\mathcal{N}_{3,3}|\leq 15$,  $\mathcal{N}_{3,3}$ is a subset of above (i) and (iii)'s arrays.
	
\end{theorem}
\begin{proof}
	
	(i) We obtain that $(1,2,3), (2,0,4)\in \mathcal{N}_{2,3}$ by \cite[Corollary $11$]{2020Inertias}. Using Lemma \ref{le:rel} we obtain that,
	$(1,0,8), (1,1,7), (1,2,6), (1,3,5), (1,4,4), (1,5,3), (2,0,7), (2,1,6),\\ (2,2,5), (2,3,4)\in \mathcal{N}_{3,3}$. We obtain that $(3,0,6)\in \mathcal{N}_{3,3}$ by Lemma \ref{le:pre}. We obtain that $(4,0,5)\in \mathcal{N}_{3,3}$ by \cite[Example $2$]{2013Negative}. We obtain that $(3,1,5)\in \mathcal{N}_{3,3}$ by Theorem \ref{th:NPT2}.
	
	(ii) We have known that $v_-\leq 4$ by \cite[Lemma $5$]{2020Inertias} and $v_+\geq 3$ by Lemma \ref{th:N_33}. We have that $v_-+v_0+v_+=9$. Then  $\mathcal{N}_{3,3}$ migtht have the following five arrays,
	$(2,4,3)$, $(3,2,4)$, $(3,3,3)$, $(4,1,4)$, $(4,2,3)$. Then we have that $(2,4,3), (3,3,3), (4,2,3)\notin \mathcal{N}_{3,3}$ by Lemma \ref{le:12} and \ref{le:13}. 
	
	(iii) From (ii), we get that $\mathcal{N}_{3,3}$ might have the following two arrays,
	$(3,2,4)$, $(4,1,4)$. 
	
	(iii) We have $13\leq |\mathcal{N}_{3,3}|\leq 15$ by the results in (i),(ii),(iii) and $ \mathcal{N}_{3,3}$ is a subset of above (i) and (iii)'s $15$ arrays.
\end{proof}
Then we propose the following example to illustrate the Theorem \ref{th:N_33}(i).
\begin{example}
	\label{ex:13arrays}
	
Suppose $\rho$ is a two-qutrit bipartite NPT state. 
	(i) Suppose $\rho=(\ket{0,0}+\ket{1,1})(\bra{0,0}+\bra{1,1})+\frac{1}{10}I_9$. Then $\ln \r^\G=(1,0,8)$. 
	
	(ii) Suppose $\rho=(\ket{0,0}+\ket{1,1})(\bra{0,0}+\bra{1,1})+\frac{1}{10}I_9-\frac{1}{10}\ket{2,2}\bra{2,2}$. Then $\ln \r^\G=(1,1,7)$.
	
	(iii) Suppose $\rho=(\ket{0,0}+\ket{1,1})(\bra{0,0}+\bra{1,1})+\frac{1}{10}I_9-\frac{1}{10}(\ket{2,2}\bra{2,2}+\ket{2,1}\bra{2,1})$. Then $\ln \r^\G=(1,2,6)$.
	
	(iv) Suppose $\rho=(\ket{0,0}+\ket{1,1})(\bra{0,0}+\bra{1,1})+\frac{1}{10}I_9-\frac{1}{10}(\ket{2,2}\bra{2,2}+\ket{2,1}\bra{2,1}-\ket{2,0}\bra{2,0})$. Then $\ln \r^\G=(1,3,5)$.
	
	(v) Suppose $\rho=(\ket{0,0}+\ket{1,1})(\bra{0,0}+\bra{1,1})+(\ket{2,2})(\bra{2,2})$. Then $\ln \r^\G=(1,4,,4)$.
	
	(vi) Suppose $\rho=(\ket{0,0}+\ket{1,1})(\bra{0,0}+\bra{1,1})$. Then $\ln \r^\G=(1,5,3)$.
	
	(vii) Suppose $\rho=(\ket{0,0}+\ket{1,1}+\ket{2,2})(\bra{0,0}+\bra{1,1}+\bra{2,2})+(\ket{0,1}+\ket{1,0})(\bra{0,1}+\bra{1,0})$. Then $\ln \r^\G=(2,0,7)$.
	
	(viii) Suppose $\rho=(\ket{0,0}+\ket{1,1}+\ket{2,2})(\bra{0,0}+\bra{1,1}+\bra{2,2})+(\ket{0,0}+\ket{0,1}+\ket{1,1})(\bra{0,0}+\bra{0,1}+\bra{1,1})$. Then $\ln \r^\G=(2,1,6)$.
	
	(ix) Suppose $\rho=(\ket{0,0}+\ket{1,1})(\bra{0,0}+\bra{1,1})+(\ket{0,0}+\ket{0,1}+\ket{2,2})(\bra{0,0}+\bra{0,1}+\bra{2,2})$. Then $\ln \r^\G=(2,2,5)$.
	
	(x) Suppose $\rho=(\ket{0,0}+\ket{1,1})(\bra{0,0}+\bra{1,1})+(\ket{0,1}+\ket{0,1})(\bra{1,2}+\bra{1,2})$. Then $\ln \r^\G=(2,3,4)$.
	
	(xi) Suppose $\rho=(\ket{0,0}+\ket{1,1}+\ket{2,2})(\bra{0,0}+\bra{1,1}+\bra{2,2})+\ket{1,1}\bra{1,1}$. Then $\ln \r^\G=(3,0,6)$.
	
	(xii) Suppose $\rho=(\ket{0,0}+\ket{1,1}+\ket{2,2})(\bra{0,0}+\bra{1,1}+\bra{2,2})+(\ket{0,0}+2\ket{0,1}+2\ket{1,1})(\bra{0,0}+2\bra{0,1}+2\bra{1,1})$. Then $\ln \r^\G=(3,1,5)$.
	
	(xiii) Suppose $\rho=(\ket{0,0}+\frac{1}{4}\ket{1,1}+\ket{2,2})(\bra{0,0}+\frac{1}{4}\bra{1,1}+\bra{2,2})+(\ket{0,1}+\frac{1}{3}\ket{1,2}+\frac{1}{3}\ket{2,0})(\bra{0,1}+\frac{1}{3}\bra{1,2}+\frac{1}{3}\bra{2,0})+(\ket{0,2}+\frac{1}{2}\ket{1,0}+\ket{2,1})(\bra{0,2}+\frac{1}{2}\bra{1,0}+\bra{2,1})$. Then $\ln \r^\G=(4,0,5)$.
\end{example}

\section{application}
\label{sec:app}
In this section, we study the relationship between the inertias of $2\times 3$ states and those of $3\times 3$ states firstly. Next we extend some conclusions on the inertias  from  $2\times n$ states to
 $3\times n$ states. We also partially test the existence of two unverified inertias using python program.

 We have obtained that $\mathcal{N}_{2,3}=\{(1,2,3),(1,1,4),(1,0,5),(2,0,4)\}$ \cite{2020Inertias}. Then we have obtained that $(1,0,8)$, $(1,1,7)$, $(1,2,6)$, $(1,3,5)$, $(1,4,4)$, $(1,5,3)$, $(2,0,7)$, $(2,1,6)$, $(2,2,5)$, $(2,3,4)$, $(3,0,6)$, $(3,1,5)$, $(4,0,5)$ $\in \mathcal{N}_{3,3}$ in Theorem \ref{th:N_33}. By Theorem \ref{th:N_33} (i), we have also obtained that $(1,0,8)$, $(1,1,7)$, $(1,2,6)$, $(1,3,5)$, $(1,4,4)$, $(1,5,3)$ are from $(1,2,3)$, and $(2,0,7)$, $(2,1,6)$, $(2,2,5)$, $(2,3,4)$ are from $(2,0,4)$. Next we consider the remaining three inertias,  $(3,0,6)$, $(3,1,5)$ and $(4,0,5)$. From Example \ref{ex:13arrays} (xi), (xii), (xiii), we propose the following example.
\begin{example}
	\label{ex:3arrays}
	Suppose $\rho$ is a $2 \times 3$ bipartite NPT state. 
	
	(xi) Suppose $\rho=(\ket{0,0}+\ket{1,1}+\ket{1,2})(\bra{0,0}+\bra{1,1}+\bra{1,2})+\ket{1,1}\bra{1,1}$. Then $\ln \r^\G=(1,1,4)$.
	
	(xii) Suppose $\rho=(\ket{0,0}+\ket{1,1}+\ket{1,2})(\bra{0,0}+\bra{1,1}+\bra{2,2})+(\ket{0,0}+2\ket{0,1}+2\ket{1,1})(\bra{0,0}+2\bra{0,1}+2\bra{1,1})$. Then $\ln \r^\G=(2,0,4)$.
	
	(xiii) Suppose $\rho=(\ket{0,0}+\frac{1}{4}\ket{1,1}+\ket{1,2})(\bra{0,0}+\frac{1}{4}\bra{1,1}+\bra{1,2})+(\ket{0,1}+\frac{1}{3}\ket{1,2}+\frac{1}{3}\ket{1,0})(\bra{0,1}+\frac{1}{3}\bra{1,2}+\frac{1}{3}\bra{1,0})+(\ket{0,2}+\frac{1}{2}\ket{1,0}+\ket{1,1})(\bra{0,2}+\frac{1}{2}\bra{1,0}+\bra{1,1})$. Then $\ln \r^\G=(2,0,4)$.
\end{example} 

   From this example, we get that $(3,0,6)$ is from $(1,1,4)$, and $(3,1,5)$, $(4,0,5)$ are from $(2,0,4)$ by Example \ref{ex:3arrays}. We conclude the above findings in the following table.
   \begin{center}
   	\label{ta:1}
   \captionof{table}{the relationship between $2\times 3$ states and $3\times 3$ states}
   \begin{tabular}{|l|c|}
   	\hline
   	$\mathcal{N}_{2,3}$ & $\mathcal{N}_{3,3}$ \\
   	\hline
   	$(1,2,3)$ & $(1,5,3)$ \\ 
   	& $(1,4,4)$ \\ 
   	& $(1,3,5)$ \\ 
   	& $(1,2,6)$ \\ 
   	& $(1,1,7)$ \\ 
   	& $(1,0,8)$ \\ 
   	\hline
   	$(1,1,4)$ & $(3,0,6)$ \\
   	\hline
   	$(2,0,4)$ & $(2,3,4)$ \\
   	 & $(2,2,5)$ \\
   	 & $(2,1,6)$ \\
   	 & $(2,0,7)$ \\
   	 & $(3,1,5)$ \\
   	 & $(4,0,5)$ \\
   	 \hline
   \end{tabular} 
	\end{center}

	 So far we have investigated the $3\times 3$ states. Next we will extend the conclusions from low dimensional states to high dimensional states. For extending the conclusions on the inertias  from $2\times n$ states to $3\times n$ states, we consider the inertias of $\mathcal{N}_{3,n}$ in the following.
	\begin{lemma}
		\label{le:3n}
		There are at least $(n-1)(2n-1)$ inertias in $\mathcal{N}_{3,n}$, i.e.,
		
		$|\mathcal{N}_{3,n}|\geq (n-1)(2n-1) $, for $\forall n\geq 2$.	
		
		Furthermore, these  $(n-1)(2n-1)$ inertias in $\mathcal{N}_{3,n}$ are as follows.
		
		$(1,3n-4-j,j+3)$, $\forall 0\leq j\leq 3n-4$,
		
		$(2,3n-6-j,j+4)$, $\forall 0\leq j\leq 3n-6$,
		
		...
		
		$(n-1,n-j,n+1+j)$, $\forall 0\leq j\leq n$.
	\end{lemma}
\begin{proof}
	 	We have obtained that $(1,2n-4,3)\in \mathcal{N}_{2,n}$ \cite{2020Inertias}.  Using Lemma \ref{le:rel} we have that $(1,3n-4-j,j+3)\in  \mathcal{N}_{3,n}$, $\forall 0\leq j\leq 3n-4$.
	 	
	 Similarly, we also have that 	
	 
	 $(2,3n-6-j,j+4)\in  \mathcal{N}_{3,n}$, $\forall 0\leq j\leq 3n-6$,
	 
	 ...
	 
	 $(n-1,n-j,n+1+j)$, $\forall 0\leq j\leq n$.
	 
	 So we obtain that $\mathcal{N}_{3,n}$ obtains at least $(3n-3)+(3n-5)+...+(n+1)$, namely $(n-1)(2n-1)$ inertias.
\end{proof}

Finally, we partially investigate the existence of inertia $(4,1,4)$ and $(3,2,4)$ in terms of numerical test by using a python program. We utilize a program package named numpy \footnote{Oliphant,T.E. A guide to NumPy (Vol.1, p.85).USA: Trelgol Publishing.(2006).} to generate a $9\times3$ real random matrix $R$ \footnote{https://machinelearningmastery.com/how-to-generate-random-numbers-in-python/(2018).}. Then we multiply $R$ by its transpose $R^T$ to obtain a random positive semi-definite matrix $M=RR^T$ \footnote{https://en.m.wikipedia.org/wiki/Definite matrixNegative-definite.2C semidefinite and indefinite matrices(2002).}. After transforming $M$ into $M^\Gamma$, the program only need to calculate the feature values of $\rm M^\Gamma$ and judge its inertia index, which turns out to be neither of the two inertias $(4,1,4)$ or $(3,2,4)$. Such a numerical test supports the conjecture that neither of the two inertias exists.




\section{Conclusion}
\label{sec:con}
 We have investigated two-qutrit EWs constructed by the partial transpose of NPT states.  Furthermore, we investigated the inertias of the two-qutrit bipartite NPT states.  One open problem is to deterimine whether $(3,2,4), (4,1,4)$ are in $\mathcal{N}_{3,3}$. In the future, we need to extend more conclusions on the inertias from low dimensional states to high dimensional states.  

\bibliographystyle{unsrt}
\bibliography{changchun=inertia}

\end{document}